\documentclass[11pt]{article}
\pdfoutput=1 

\usepackage{authblk}
\usepackage{graphicx}
\usepackage{amsmath,amsthm,amssymb,amsfonts,setspace}
\usepackage{bm}
\usepackage{mathtools}
\usepackage{braket}
\usepackage{float}
\usepackage{color}
\usepackage[usenames,dvipsnames]{xcolor}
\usepackage{hyperref}  
\usepackage{cleveref} 
\usepackage{wrapfig}
\usepackage{tikz}
\usepackage[left=1in,right=1in,top=1in,bottom=1in]{geometry}
\usepackage{thmtools, thm-restate}
\usepackage{makecell}
\usepackage{thmtools}
\hypersetup{colorlinks=true,urlcolor=[rgb]{0,0,0.5},citecolor=[rgb]{0,0.5,0},linkcolor=[rgb]{0,0.5,0}}

\usepackage{graphicx}
\graphicspath{ {./Images/}}

 \theoremstyle{plain}
 
 \theoremstyle{plain}
 \newtheorem{lem}{Lemma}
 \newtheorem*{lem*}{Lemma}
 \theoremstyle{plain}
 \newtheorem{thm}{Theorem}
 \newtheorem*{thm*}{Theorem}
 \theoremstyle{plain}
  \newtheorem*{mainres}{Main result}
 
   \theoremstyle{plain}
 
 \theoremstyle{plain}
 
  \newtheorem*{corr*}{Corollary}
 \theoremstyle{plain}
 
 \theoremstyle{remark}
 \newtheorem*{rem*}{Remark}
	\theoremstyle{remark}
	 
  \theoremstyle{plain}
 \newtheorem{defn}{Definition}	 
   \theoremstyle{plain}
 \newtheorem{conj}{Conjecture}
    \theoremstyle{plain}
 
    \theoremstyle{plain}
 \newtheorem{problem}{Problem}
  \newtheorem*{problem*}{Problem}
   \newtheorem*{conj*}{Conjecture}

\newcommand{\Per}{\operatorname{Per}}
\newcommand{\Haar}{\mathrm{Haar}}

\newcommand{\superap}{\lvert \mathrm{SUPER}\rvert^2_\pm}
\newcommand{\gsuperap}{\lvert \mathrm{GSUPER}\rvert^2_\pm}
\newcommand{\supermp}{\lvert \mathrm{SUPER}\rvert^2_\times}

\newcommand{\Smn}{\mathcal{S}_{m,n}}
\newcommand{\Scmn}{\mathcal{S}^c_{m,n}}

\newcommand{\unifmn}{\mathcal{U}_{m,n}}
\newcommand{\Dmn}{\mathcal{D}_{m,n}}
\newcommand{\Gmn}{\mathcal{G}_{m,n}}

\newcommand{\Complex}{\mathbb C}

\newcommand{\eps}{\epsilon}

\renewcommand{\exp}{\mathrm{exp}}

\DeclareMathOperator{\vol}{Vol}

\DeclareMathOperator{\Rept}{Re}

\newcommand{\diff}{\mathop{}\!\mathrm{d}}
\newcommand{\Expec}{\mathbb E}
\newcommand{\Prob}{\Pr}
\newcommand{\Sphere}{{\mathbf S}}
\newcommand{\One}{{\mathbf{1}}}
\newcommand{\C}{\mathbb{C}}

\newcommand{\sharP}{\#\mathsf{P}} 
\renewcommand{\ket}[1]{\left| #1 \right>}

\renewcommand{\S}{\mathcal{S}}

\newcommand{\norm}[1]{\left\lVert#1\right\rVert}

\newcommand\poly[1]{\mathrm{poly}(#1)}

\newcommand{\normZ}{\mathcal{Z}}

\crefformat{equation}{Eq.~(#2#1#3)} 
\Crefformat{equation}{Equation~(#2#1#3)}
\crefformat{figure}{Fig.~#2#1#3}
\crefrangeformat{equation}{Eqs.~#3(#1)#4--#5(#2)#6}

\begin{document} 	
 \title{Complexity-theoretic foundations of BosonSampling with a linear number of modes}
 \date{}
\author[1]{Adam Bouland}
\author[2]{Daniel Brod}
\author[1]{Ishaun Datta}
\author[3]{Bill Fefferman}
\author[4,5]{Daniel Grier}
\author[6]{Felipe Hern\'{a}ndez}
\author[7]{Micha\l{} Oszmaniec}

\affil[1]{Department of Computer Science, Stanford University}
\affil[2]{\normalsize{Instituto de F\'{i}sica, Universidade Federal Fluminense}}
\affil[3]{Department of Computer Science, University of Chicago}
\affil[4]{Department of Computer Science and Engineering, UC San Diego}
\affil[5]{Department of Mathematics, UC San Diego}
\affil[6]{Department of Mathematics, MIT}
\affil[7]{Center for Theoretical Physics, Polish Academy of Sciences}

\maketitle	
\begin{abstract}
    BosonSampling is the leading candidate for demonstrating quantum computational advantage in photonic systems. While we have recently seen many impressive experimental demonstrations, there is still a formidable distance between the complexity-theoretic hardness arguments and current experiments.  One of the largest gaps involves the ratio of {particles} to modes---all current hardness evidence assumes a {``dilute''} regime in which the number of linear optical modes scales at least quadratically in the number of {particles}.  By contrast, current experiments operate in a {``saturated''} regime with a linear number of modes.  In this paper we bridge this gap, bringing the hardness evidence for experiments in the saturated regime to the same level as had been previously established for the dilute regime.  This involves proving a new worst-to-average-case reduction for computing the Permanent which is robust to both large numbers of row repetitions and also to distributions over matrices with correlated entries. {We also apply similar arguments to give evidence for hardness of Gaussian BosonSampling in the saturated regime.}
\end{abstract}

\section{Introduction}\label{sec:intro}

In the decade since it was proposed by Aaronson and Arkhipov \cite{Aaronson2013}, BosonSampling has become one of the most promising candidates for achieving quantum computational advantage---the experimental demonstration of a quantum computation which exponentially surpasses classical computers. This requires a task which is both experimentally feasible and has strong complexity-theoretic evidence for hardness. We have now seen several experimental demonstrations of BosonSampling \cite{wang2019boson,young2024atomic} and its Gaussian variant at scale \cite{GaussBSExperiment2020, zhong2021phase, madsen2022quantum}, as well as substantial work building the theory of these experiments and bringing them closer to feasibility (see e.g.\ \cite{hamilton2017gaussian,chakhmakhchyan2017boson,deshpande2021quantum,grier_brod_2021}).

Despite this progress, there is still a formidable distance between experiment and theory. One of the most notable gaps involves the ratio between the number of modes and number of {particles (typically photons)}.  The original BosonSampling proposal calls for a ``dilute regime'' in which $n$ photons are passed through an $m=\Omega(n^2)$-mode Haar-random interferometer followed by the measurement of each mode in the particle-number basis.\footnote{The original paper required $m=\omega(n^5)$ modes, but showed that $m=\Omega(n^2)$ would suffice under a plausible random matrix theory conjecture. Subsequent work proposed a variant, known as Bipartite Gaussian BosonSampling \cite{grier_brod_2021}, that works with $\Omega(n^2)$ modes. } By contrast, all experiments to date operate in a ``saturated regime'' where the number of modes is comparable to the number of photons i.e., $m=\Theta(n)$. {The ratio $m/n$ varies from 3 to 5.6 in recent BosonSampling experiments \cite{wang2019boson,young2024atomic} and from  0.5 to 2 in recent Gaussian BosonSampling experiments \cite{GaussBSExperiment2020,zhong2021phase,madsen2022quantum,GaussBSExperiment2023}}.  Understanding this saturated regime  has long been cited as a major open problem going back to the original paper \cite{Aaronson2013} which asked explicitly:
\begin{center}
\begin{minipage}{.8\linewidth}
\begin{center}
\textit{``Can we reduce the number of modes needed for our linear-optics experiment, perhaps from $O(n^2)$ to $O(n)$?''}
\end{center}
\end{minipage}
\end{center}
It is reasonable to conjecture that the need for such a high number of modes is an artifact of current proof techniques, rather than intrinsic to the hardness of the sampling problem.  For one, state-of-the-art classical simulation algorithms are not able to take advantage of the saturated regime to achieve dramatically faster runtimes
\cite{clifford2018classical}, albeit some improvements are possible \cite{clifford2020faster,newmichalpaper}.  {Furthermore, 
the saturated regime is sufficient to perform universal quantum computation \cite{knill2001scheme} with single-photon sources, linear optics and (adaptive) photon-detectors}. 

\begin{figure}[t]
\centering
    \includegraphics{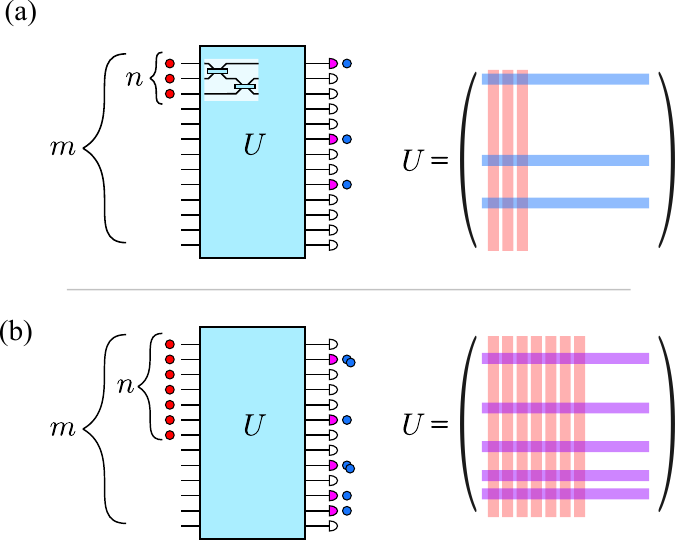}
\caption{Two regimes of BosonSampling. (a) The dilute limit, where $m=\Omega(n^2)$. No-collision outcomes dominate, and the submatrix whose permanent we must compute (i.e.\ selecting the appropriate rows and columns) is sufficiently smaller than the Haar-random interferometer $U$ that its entries look i.i.d.\ Gaussian. (b) The saturated limit, where $m=\Theta(n)$. No-collision outcomes are rare, and complexity is determined by the number of detector clicks (purple) rather than  output photons (blue). The entries of the submatrix of $U$ are very correlated, and its rows will be repeated according to collisions, as described in the main text.}
\label{fig:setup}
\end{figure}

Analyzing the hardness of BosonSampling in the saturated regime is quite challenging for two reasons {(see Figure \ref{fig:setup} for a graphical presentation)}. {The first is due to collisions, i.e., events where more than one photon is detected in a given output mode. Current proof techniques rely heavily on a property of the dilute regime known as the ``Bosonic Birthday Paradox'' \cite{arkhipov_kuperberg} which ensures that most measurement outcomes are collision-free}.  By contrast, BosonSampling in the saturated regime has a large number of collisions, {which yields many inequivalent output types that have to be independently analyzed}.  Second, the probability of each outcome is the squared permanent of a submatrix of the unitary that encodes the interferometer, which is chosen to be Haar random. In the dilute limit these submatrices have entries that are approximately i.i.d.\  Gaussian, which are convenient to analyze, whereas in the saturated limit, the relevant submatrices do not have i.i.d.\ entries \cite{jiang2006many}.

In this work we overcome these obstacles and build the complexity-theoretic foundations for BosonSampling in the saturated regime, answering Aaronson and Arkhipov's question in the affirmative. The starting point is to identify typical outputs that occur in generic BosonSampling experiments in the saturated regime. {Specifically, we prove that typical outcomes in $n$-particle and $m$-mode BosonSampling experiments  are characterized by  $c\approx \frac{m/n}{m/n +1} n$ \emph{clicks} (i.e., modes occupied by at least a single particle).  Therefore, in the saturated regime we still have $c=\Theta(n)$. Then, using methods from \cite{grier_brod_2021} we show that probabilities of such typical outcomes can be used to encode permanents of square matrices of size $c$. This connection allows us to apply the complexity-theoretic framework of Aaronson and Arkhipov  to prove that hardness of classical approximate sampling from BosonSampling in the experiments follows from the hardness of an appropriate average-case hardness conjecture:}

\begin{mainres}[Informal]
    Assuming average-case hardness of approximating {random} output probabilities of BosonSampling experiments  for {$m\geq 2.1 n$}, there is no efficient classical algorithm to sample from the output distribution of such an experiment. Furthermore, we prove a weaker version of this average-case hardness conjecture. The gap between what is proven and what is conjectured is the same as for standard BosonSampling in the dilute {regime}.
\end{mainres}

\begin{rem*}
    While for concreteness we state our results for $m>2.1n$ throughout this paper, in fact all such statements pertain more generally to $m = \lceil(2+\eps)n\rceil$ and $n\geq 2\eps^{-1}$ for any $\eps>0.$
\end{rem*}
Our main results give strong evidence in favor of this average-case hardness conjecture. In particular we show it is \textsf{\#P}-hard to 
{approximately} compute {random} output probabilities of random experiments in the saturated {regime, albeit to a smaller error than would be necessary}. We also provide numerical evidence for anticoncentration\footnote{Informally, anticoncentration is the property that output probabilities are not too concentrated around zero. It remains a conjecture for BosonSampling variants, though it is considered a ``milder'' conjecture than average-case hardness.} in this regime. 
This brings the saturated regime of BosonSampling to essentially the same level of theoretical support as dilute-limit experiments. {Importantly, we also apply our complexity-theoretic analysis to Gaussian BosonSampling, which is the choice of state-of-the-art experiments \cite{GaussBSExperiment2020,zhong2021phase,madsen2022quantum,GaussBSExperiment2023}. In this case the saturated regime where our results formally hold (corresponding to $m> 2.1, n$) match to a squeezing parameter of 0.65, which is well within the range of experimentally accessible values \cite{zhong2021phase}.}

{Beyond reinforcing the theoretical underpinnings of saturated-regime of BosonSampling, our work also has consequences for experimental demonstrations of BosonSampling. Specifically, a combinatorial characterization of number of detector cliques in the saturated regime can serve as a benchmark for current and future experiments which operate in the saturated regime. From a physics perspective, our analysis highlights the role of \emph{number of detector clicks} as the measure of hardness of BosonSampling experiments beyond the diluted regime. 
}


\subsection{Proof Sketch} \label{sec:sketch}

{We begin by reviewing the well-known connection between linear-optical experiments and the permanent function \cite{scheel,Aaronson2013}. This will serve as background for the unfamiliar reader, but also to lay out  more concretely how the new challenges which we attributed to the saturated regime arise.} 

{Suppose we have $m$ photonic modes on which we prepare a Fock state of $n < m$ photons. Each photon is initialized in a different mode, and without loss of generality we assume that the first $n$ modes are initially occupied. These modes are then subject to a linear-optical evolution (or interferometer), i.e., a transformation of the type
\begin{equation}
    a_j^\dag \longmapsto \sum_i U_{ij} a_i^\dag,
\end{equation}
where $a_j^\dag$ is the bosonic creation operator corresponding to mode $j$ and $U$ is an $m \times m$ unitary matrix. Finally, we detect the number $s_i$ of photons that exit on the $i$th output mode of the interferometer. The transition probability from input $\ket{S_0} = \ket{1, 1, \ldots 1, 0, \ldots 0}$ to output $\ket{S} = \ket{s_1, s_2 \ldots s_m}$ is given by
{
\begin{equation} \label{eq:probability0}
    p_U(S) := \frac{\left|\Per(U_S)\right|^2}{\prod_{i=1}^m s_i!}\ ,
\end{equation}
}

{where $U_S$ is an $n \times n$  matrix constructed by taking the leftmost $n$ columns of $U$}, and then subsequently taking $s_i$ copies of the $i$th row.}

{Equation (\ref{eq:probability0}) is the core connection between linear-optical experiments and the permanent function. It showcases more concretely the new challenges of the saturated regime: (i) the lack of a Bosonic Birthday Paradox implies that, for typical $U$, collisions occur frequently in the output probability distribution. 
This in turn implies that the relevant matrices $U_S$ almost surely have repeated rows; and (ii) the fact that $U_S$ is a relatively large submatrix of $U$, and we cannot argue that its entries look i.i.d. Gaussian when $U$ is uniformly random. Both challenges break major proof steps of the original BosonSampling proof \cite{Aaronson2013} which we need to adapt, so now we discuss each in turn. }

The first obstacle---the presence of photon collisions in typical outcomes---breaks a property known as \emph{hiding}. Hiding
is the property that all outputs of the experiment are symmetrical over the choice of random experiment. This simple property plays a surprisingly key role in current quantum advantage arguments. The basic reason is that these arguments try to show no approximate classical sampler exists to small total variation distance error {(over the full  outcome space)}. If all outputs are on equal footing then this error can be spread over all outputs by Markov's inequality. However, if only a few outputs matter for the hardness arguments, one might worry the approximate sampler could corrupt these outputs only{---producing a distribution that is close in total variation distance to the ideal one while containing no information about the actually hard outcomes.} This symmetry property trivially holds for Random Circuit  Sampling \cite{bouland_complexity_2019}, IQP \cite{bremner2016iqp}, Fermion Sampling \cite{fermionsampling}, and many other advantage schemes, as well as  for BosonSampling in the dilute regime. However, in the saturated regime this symmetry fails spectacularly---the output space shatters into an exponential number of incomparable output types each with probability mass that is relatively well spread. 

We instead proceed by formulating a modified version of the Stockmeyer counting reduction \cite{stockmeyer1983complexity,Aaronson2013} which does not require the hiding symmetry.  We choose a uniformly random outcome of the experiment and use Stockmeyer counting to estimate the probability of this outcome. 
By Markov's inequality we can ensure that most output probabilities of the approximate sampler are approximately correct. However, this modified reduction comes at a cost---to show hardness of sampling, it no longer suffices to show hardness of computing  output probability $p_U(S)$ for a single type of $S$, but instead we must now show an \emph{entire suite} of hardness results for \emph{most} outputs $S$ of the experiments. More formally, in the saturated limit the output space consists of an exponential number of incomparable collision patterns---i.e., unordered lists of occupation numbers of modes. We need to argue it is hard to estimate the output probabilities {$p_U(S)$ for \emph{typical} $S$ chosen from the uniform measure in the set of possible outcomes}. 

Our next step is to show such a suite of average-case hardness results of computing most outputs of {BosonSamplimg experiments in the saturated regime}:
\begin{thm}[Informal] \label{thm:maininformal} It is $\sharP$-hard to compute most output probabilities of most BosonSampling experiments  in the saturated regime to within additive error $e^{-O(n\log n)}$. 
\end{thm}

This is nearly what we need to show hardness of sampling via the modified Stockmeyer reduction. {By adaptating the results of \cite{Bouland2021} to our case, it can be shown that robustness to additive error $e^{-O(n)}$ would remove the need for additional hardness conjectures.}
Here we need to overcome both of the major differences that distinguish the dilute and saturated regimes.  

{First, one must deal with the presence of collisions. To do this, we identify a collection of output types which cover a large fraction of the output probability distribution of typical experiments. This requires a careful combinatorial accounting of typical collision patterns in typical outputs of experiments in the saturated limit. In particular, our argument suggests that the main figure of merit of complexity in the saturated regime is not the number of photons or modes in the experiment, but rather the number of output detector clicks (i.e., output modes which receive at least one photons), which informally follows from the fact that, if we only use O$(1)$ distinct rows repeated many times, the resulting $n \times n$ matrix has very low rank and its permanent is in fact easy to compute. Thus, we need to show that, in a typical experimental run, the number of clicks is still large enough even if collision-free states are rare.} 

Once a suitable collision pattern is identified, we need to show average-case hardness for computing outputs of that collision type over the random choice of interferometer. This is equivalent to showing hardness of computing the permanent of a large submatrix of a random unitary, with a particular pattern of repeated rows. 

{In the dilute regime this average-case hardness argument proceeds as follows (using a variant of Lipton's argument \cite{lipton1991}). It starts by arguing that a Gaussian matrix is perturbed only slightly by shifting and rescaling, i.e., if $B$ is a matrix of i.i.d.\ Gaussian elements, then $(1-t) B + t A$ is also approximately Gaussian under some mild conditions over $A$ and for sufficiently small $t$. The polynomial nature of the permanent function implies then that we can use polynomial interpolation to estimate the permanent of $A$ from a number of estimates of the permanent of $A$.}  Put more simply, we can, in some sense, ``sneak'' a tiny amount of a worst-case matrix into an average-case matrix. Since this proof uses an entry-by-entry analysis of the matrix, it breaks in the saturated case since the relevant submatrices are far from i.i.d.\ Gaussian \cite{jiang2006many}---instead, the entries come from a highly correlated measure.
We show that, perhaps surprisingly, this highly correlated distribution is nonetheless also approximately shift-and-scale invariant so that we recover Lipton's proof. To do this, we directly study the probability density of the singular values of a submatrix of a Haar-random unitary \cite{collins2003integrales,reffy}.
We reduce the desired invariance property to estimating the gradient of this probability density, which we show is equivalent to proving sharp tail bounds on the maximum singular value. Our desired bounds require that we go beyond generic concentration inequalities such as Levy's lemma or log-Sobolev inequalities, and instead we derive them from high-dimensional geometry. Considering the ubiquity of the Haar measure over unitaries, we expect that this bound may be of independent interest.

Finally we address the issue of anticoncentration of {outcome probabilities} of saturated experiments. Anticoncentration is a necessary ingredient for converting additive estimates of output probabilities to multiplicative estimates---which we conjecture to be hard. It remains open to prove anticoncentration for all variants of BosonSampling, but there has been partial progress in this direction \cite{nezami2021permanent}, as well as  numerical evidence for anticoncentration in the dilute regime \cite{Aaronson2013}. However, as one reduces the number of modes, one might worry that anticoncentration might begin to fail, both due to the row repetitions in the submatrices, and the correlations between submatrix entries. This could be an issue as many of the known attacks on quantum advantage schemes hold in the non-anticoncentration regime, e.g.\ constant-depth random circuits \cite{Napp2020}. To alleviate this concern and to support our anticoncentration conjecture, we provide numerical evidence that anticoncentration holds in the saturated regime, and indeed has very similar behavior to the i.i.d.\  Gaussian case. 

{This paper is organized as follows. In Section \ref{sec:notation} we define the notation used throughout the paper. In Section \ref{sec:classical-hardness} we prove that efficient classical simulation of saturated-regime BosonSampling is not possible up to some complexity-theoretic conjectures. In particular, we need to formulate a new conjecture adapted to the computational problem that arises in this regime.   {In Section \ref{sec:manyclicks} we present an analysis of the combinatorics of collisions in typical BosonSampling experiments}. In Section \ref{sec:totalaveragehardness} we provide evidence towards our conjecture.
{by proving a weaker version of the desired result}. In Section \ref{sec:anticonc} we provide numerical evidence for anticoncentration of the relevant output probabilities. In Section \ref{sec:GBS} we discuss how our results can be adapted to the case of Gaussian BosonSampling {in the saturated regime}. Finally, in Section \ref{sec:conclusions} we lay out some open questions and final remarks.}

\section{Notation} \label{sec:notation}
In this section we collect notation used throughout the paper.
We use the letter $n$ to denote the number of photons and $m$ the number of modes.  
We use $\mathcal S_{m,n}$ for the set of possible outcome patterns $S=(s_1,\ldots,s_m)$ where $s_j$ denotes the number of photons measured in the $j$-th mode.  More precisely, we define
\begin{equation}
\Smn
= 
\{(s_1,s_2,\ldots,s_m) \mid  s_j\in\mathbb N_{\geq 0} \text{ and } \sum_{j=1}^m s_j = n\}.
\end{equation}
{There is a one-to-one correspondence between elements of $\Smn$ and Fock states of $n$ bosons occupying $m$ modes. Therefore, the cardinality of $\Smn$ equals the dimension of the Hilbert space $\mathrm{Sym}^n(\C^m)$ spanned by these states.}   By a stars-and-bars counting argument, the cardinality of the set
$\mathcal{S}_{m,n}$  is
\begin{equation}
|\mathcal{S}_{m,n}| = \binom{m+n-1}{n}.
\end{equation}
{In most of the paper}  we assume the standard initial state, i.e.\ $S_0=(1,\ldots,1,0,\ldots,0)$ with the first $n$ modes occupied each by a single boson.
The number of nonzero entries in a collision pattern $S$ is referred to as the number of \textit{clicks} and we will donote it by $c_S$ {(this is a reference to a type of detector, known as a bucket detector \cite{brod2019photonic}, which literally clicks when one or more photons arrive at an output mode, but cannot distinguish the actual photon number)}.  We will also denote by $k_S = n-c_S$ the number of excess particles, which for simplicity we refer to as collisions.

We use $\unifmn$ to refer to the uniform distribution on $\mathcal{S}_{m,n},$ i.e.\ the distribution such that every outcome $S\in\mathcal{S}_{m,n}$ has probability $1/\vert\Smn\vert$. {We will denote by $p_U$ the probability distribution of $\Smn$ induced by BosonSampling experiments.
For unitary $U$ and outcome $S$, we define matrix  $U_S$ by taking the $m \times n$ left submatrix of $U$ and repeating the $i$th row $s_i$ many times. The probability distribution $p_U(S)$ on $\mathcal{S}_{m,n}$ is given by}
{
\begin{equation} \label{eq:probability}
    p_U(S) := \frac{\left|\Per(U_S)\right|^2}{\prod_{i=1}^m s_i!}\ .
\end{equation}
}
Finally, we use $\Dmn$ to refer to the probability distribution on $n\times n$ matrices $U_S$ induced by drawing $U\sim \Haar(m)$ and $S\sim\mathcal U_{m,n}$. 

We also heavily rely on the standard $O$, $\Omega$ and $\Theta$ notation. For two functions $f(n)$ and $g(n)$, we say that $f$ is O$(g)$ if $f(n)$ is upper-bounded by some positive multiple of $g(n)$ in the asymptotic limit $n \rightarrow \infty$. We say that $f$ is $\Omega(g)$ if $g$ is $O(f)$, and we say $f$ is $\Theta(g)$ if it is both O$(g)$ and  $\Omega(g)$.

\section{Efficient classical simulation of BosonSampling in the saturated regime collapses \textsf{PH}}\label{sec:classical-hardness}

In this section, we argue that efficient classical simulation of BosonSampling in the $m=\Theta(n)$  regime would imply a collapse of the Polynomial Hierarchy, by a reduction from sampling (the experiment's distribution over outcomes) to computing (outcome probabilities). 
Our proof is based on an average-case hardness conjecture for a problem we call the Sub-Unitary Permanent Estimation with Repetitions, or $\mathrm{SUPER}.$
In subsequent sections, we provide evidence for the hardness of this problem. 

The key ideas behind the hardness proof \cite{Aaronson2013}  are as follows: imagine there exists a classical sampler for a saturated BosonSampling experiment, namely a probabilistic polynomial-time algorithm that produces an outcome $S$ from a distribution close in $l_1$ distance to the experiment's. By Markov's inequality, most of the sampler's estimates to the true outcome probabilities are reasonably accurate. However, a celebrated result due to Stockmeyer, known as the approximate counting method, estimates those very outcome probabilities in $\mathsf{BPP}^\mathsf{NP}$ (i.e., in probabilistic polynomial-time with access to an \textsf{NP} oracle) \cite{stockmeyer1983complexity}. Stockmeyer's method exploits the fact that, unlike for a quantum computation, the classical sampler can be treated as a deterministic algorithm that takes a random input. {Combining the fact that $\mathsf{BPP}^\mathsf{NP}$ lies within the Polynomial Hierarchy (PH), our conjecture that it is $\sharP$-hard to estimate probabilities of most outcomes from most  experiments in the saturated limit, and Toda`s theorem (which informally place $\sharP$ outside of PH)}, we conclude that the existence of a classical sampler collapses \textsf{PH}--violating a bedrock assumption of complexity theory. 

Why introduce a new conjecture to argue for the classical intractability of saturated BosonSampling? As described in Section \ref{sec:intro}, the need for a new basis to demonstrate hardness of sampling comes from the two major differences between the dilute and saturated settings, namely collisions and correlations. In the saturated regime, outcome probabilities are not given by permanents of submatrices of Haar unitaries as in the dilute regime; instead, the matrices of interest have repeated rows due to collisions. Moreover, regardless of collisions, the relevant matrices in saturated BosonSampling manifestly do not have approximately i.i.d.\ Gaussian entries. Row repetitions and correlated entries disconnect standard assumptions in the BosonSampling literature, such as the Permanent-of-Gaussians Conjecture, from the saturated regime. 
To this end, we define the following problem about estimating average-case output probabilities. 

\begin{problem}[Sub-Unitary Permanent Estimation with Repetitions, $\superap$]
\label{prob:super}
We define a family of problems parametrized by error scaling function $\varepsilon \colon \mathbb N^m \to \mathbb R^+$ and failure probability $\delta$. Recall that $\Dmn$ is the distribution on $n\times n$ matrices $V$ obtained by drawing independently $U\sim\Haar(m)$ and $S\sim\unifmn$, and setting $V=U_S$. Given $V\sim\Dmn$ and error probability $\delta>0,$ output $z\in\mathbb{C}$ such that $|z - |\Per(V)|^2 | \le \varepsilon(S)$
with probability at least $1 - \delta$ over choice of $V$\footnote{Note that $\varepsilon(S)$ is well-defined as it can only depend on the multiplicities of rows of $V$.}.
\end{problem}

We conjecture that $\superap$ is $\sharP$-hard, in particular:

\begin{conj}\label{conj:UPEap}
$\superap$ is $\sharP$-hard to additive imprecision $\varepsilon(S)=\frac{1}{\poly{n}}\prod_i s_i!/\vert\Smn\vert$, when  ${m\geq 2.1 n}$ {with failure probability $\delta\leq\frac{1}{4}$}.
\end{conj}

Here, $\prod_i s_i!$ is a normalization factor inherited from bosonic statistics; see \cref{eq:probability}. By contrast, in dilute BosonSampling there is only one output type and this term is trivially 1.

\Cref{thm:no-sampler} (stated informally in \Cref{sec:intro} as part of the main result) shows that \Cref{conj:UPEap} implies quantum computational advantage for $m=\Theta(n)$ BosonSampling. 

\begin{thm}[No classical sampler]\label{thm:no-sampler}
    Let $m\geq 2.1 n$. An existence of a classical sampler that samples from a distribution approximating the BosonSampling distribution (cf.\ Eq.\ \eqref{eq:probability}) to accuracy $\beta=1/\poly{n}$ implies a solution to $\superap$ with additive imprecision ${\epsilon(S)}=\frac{1}{\poly{n}}\prod_i s_i!/\vert\Smn\vert$ in $\mathsf{BPP}^\mathsf{NP}. $ Assuming Conjecture \ref{conj:UPEap}, this collapses $\mathsf{PH}$.
\end{thm}

We modify the usual Stockmeyer reduction to choose a uniformly random outcome. This subtlety does not arise in dilute BosonSampling, where up to permuting the modes all outcomes are the same. The proof is deferred to Appendix \ref{sec:stockmeyer}.

\section{Random measurement outcomes have many occupied modes}\label{sec:manyclicks}

{For our hardness proofs it is crucial to observe that in the saturated regime bosons do not bunch too much. That is if $m=\Theta(n)$ for typical BosonSampling experiments there are $\Theta(n)$ occupied modes (clicks) upon measurement.
This property is important because the bosonic transition amplitudes are proportional to permanents of matrices with repeated rows according to the collision profile. The rank of this matrix is at most the number of clicks. For constant number of clicks one can use efficient algorithms for constant rank matrices, e.g.\ those due to Barvinok and to Ivanov and Gurvits \cite{barvinok1996two, Ivanov2020}.  Furthermore there exists an efficient classical algorithm to compute the permanent of an $n\times n$ matrix with only $O(\log(n)$ distinct rows (or columns) and only a constant number of rows (columns) are repeated  \cite{Brod2020classicalsimulation,chin2018generalized2018}. Therefore, in these high-collision regimes transition amplitudes are easy to compute. That this is \textit{not} the case for $m=\Theta(n)$ BosonSampling is the content of this section. Lemma \ref{lem:manyclicks} shows that $m/n=\alpha$   BosonSampling is dominated by outcomes with at least ${(\frac{\alpha}{\alpha+1})}n$ clicks. As we will see in Section \ref{sec:totalaveragehardness} this is enough to ensure that the matrices induced by collisions in this regime inherit the hardness of computing permanents for almost all outcomes. }

{We will be interested in the distribution of clicks $c$ for outcomes $S$ distributed according to the expected probability distribution of BosonSampling experiments $p_{av}(S)=\mathbb{E}_Up_U(S)=1/|\Smn|$ (note that $p_{av}$ is in fact the uniform distribution over $\Smn$).} We will use $\mathcal{S}^c_{m,n}$ to denote the subset of $\mathcal{S}_{m,n}$ of outcomes with exactly $c$ clicks, i.e., $c$ occupied modes, $c \coloneqq \lvert \{s_i \mid s_i\geq 1\}\rvert $, where $c\in\{1,\ldots,\min(m,n)\}.$ Observe $\mathcal{S}_{m,n} = \bigcup_{c} \mathcal{S}^c_{m,n}.$

\begin{lem}[Combinatorics of $c$ clicks]\label{lem:states-with-c-clicks}
    For $n,m \in \mathbb{N}$ and $c\in\{1,\ldots,\min(m,n)\},$ \[\lvert \Scmn \rvert = \binom{m}{c} \binom{n-1}{n-c}.\] Consequently, the probability that a randomly chosen $S\sim \Smn $ will have exactly  $c$ cliques with probability
    \begin{equation}\label{eq:hypergeom}
   p_{m,n}(c)= \frac{ |\Scmn|} {|\Smn|} =  \frac{\binom{m}{c} \binom{n-1}{n-c}}{\binom{m+n-1}{n}}\ .
\end{equation}
\end{lem}

\begin{proof}
The number of states with $c$ clicks are enumerated by first choosing which $c$ of the $m$ modes to ``click,'' i.e.\ allotting one boson each to $c$ modes from among $m.$ This contributes a factor of $\binom{m}{c}.$ Next, we allot the remaining $n-c$ bosons to the $c$ selected modes. The factor for this contribution equals $|S_{c,n-c}|$ the dimension  is given by the standard stars-and-bars argument as $\binom{(n-c)+ c - 1}{n - c} = \binom{n-1}{n-c}$.
\end{proof}

{To establish the typical behavior of $c_S$ when $S\sim\Smn$ we observe that $p_{m,n}(c)$ is in fact the hypergeometric distribution $H(M,N,n,c)$ with parameters $M=m$, $N=m+n-1$, $n= n$ and $c\in\{1,\ldots,\min(m,n)\}$ (we use the notation from \cite{hypergeometricTAIL}). The mean of $c$ distributed in this manner equals $\mathbb{E}[c]=\frac{M}{N} n = \frac{\alpha}{\alpha+1+1/n} n$, where $\alpha=m/n$.  Using large deviation bounds established in \cite{hypergeometricTAIL} we get strong quantitative bounds on the probability that $c_S$ is far form the average.  
\begin{lem}[Concentration inequality for number of occupied modes for random] \label{lem:manyclicks}
Let $c_S$ be a number of clicks in a random outcome $S\sim\unifmn$ (i.e. $c_S$ is distributed according to distribution $p_{m,n}$ in Eq.~\eqref{eq:hypergeom}). Then, for $t\geq0$ and $\alpha=m/n$ we have
\begin{equation}\label{eq:largeDEVIATIONcliques}
    \Prob_{S\sim\unifmn}\left[ \left|c_S - \frac{\alpha}{\alpha+1 +1/n} n\right| \geq t n    \right]\leq 2\ \exp(-2 t^2 n)\ .
\end{equation}
\end{lem}
This shows that states with $\Theta(n)$ occupied modes dominate in the $m=\Theta(n)$ regime.
}

\section{Average-case hardness of saturated BosonSampling probabilities} \label{sec:totalaveragehardness}

In this Section, we prove a weaker version of \Cref{conj:UPEap}, with a smaller error tolerance. We begin in \Cref{sec:worst-case} by proving a worst-case hardness result for permanents of binary matrices with repeated rows. In \Cref{sec:avgcase} we leverage that result and a polynomial-interpolation argument to prove average-case hardness.

\subsection{Worst-case hardness of permanents with repetitions} \label{sec:worst-case}

In \Cref{sec:classical-hardness}, we introduced Sub-Unitary Permanent Estimation with Repetitions ($\mathrm{SUPER}$). In this section, we show that even when we relax the unitarity constraint and consider simply submatrix permanents with repetitions, the problem is still $\sharP$-hard. That is, here we demonstrate hardness of computing permanents of arbitrary (i.e.\ \textit{worst-case}) {binary} submatrices with average-case repetitions {discussed in Section \ref{sec:manyclicks} }\footnote{Whereas in dilute BosonSampling there is only one type of collision pattern, namely $n$~1's and $m-n$~0's, recall from Section \ref{sec:intro} that in the saturated limit the output space shatters into exponentially many incomparable output types that have well-spread probability mass. The consequence is that for  BosonSampling in the saturated limit, we always have to consider statements for \textit{most} outcomes $S$, namely average-case repetitions.}. This plays an integral role in proving the hardness of computing the permanent of \textit{average-case} submatrices with average-case repetitions, as it forms the base problem of our worst-to-average-case reduction in Section \ref{sec:avgcase}.

\begin{thm}[$\sharP$-hard permanents of worst-case submatrices]\label{thm:worst}
   { Let $m=\Theta(n)$ and $S\sim\unifmn$. Let $c_S$ denote the number of clicks in $S$. Then with probability at least $1-\exp(-\Theta(n)$ there exist a poly-time constructable  matrix $A \in \{0, 1\}^{c_S \times n}$, such that it is $\sharP$-hard to compute $\Per(A_S)$, where $A_S$ is obtained from $A$ by repeating its rows according to collision pattern $S$.}   
\end{thm}

{The proof of Theorem \ref{thm:worst} uses Lemma \ref{eq:largeDEVIATIONcliques}, which guarantees that with high probability $c_S= \Theta(n)$ once $m=\Theta(n)$, and the following reduction between the problem of computing of permanents of binary matrices with repeated rows and the problem of computing permanents of general binary matrices.}

\begin{lem} 
\label{lem:rep_embedding}
Let $m\geq n$. Given collision pattern $S\in\Smn$ with $c$ clicks and matrix $X \in \{0, 1\}^{c \times c}$, there is a poly-time constructible matrix $A \in \{0, 1\}^{c \times n}$ such that 
\begin{equation}
\Per(A_S) = \Per(X) \prod_{i=1}^c s_i!
\end{equation}
\end{lem}

\Cref{lem:rep_embedding} follows directly from Lemma 24 of \cite{grier_brod_2021}. The general idea is as follows. We start with $c \times c$ matrix $X,$ which encodes some $\sharP$-hard problem. We append to it a $c \times k$ matrix $Y$ to make a larger $c \times (k+c)$ matrix, $A$. Finally, we repeat $k$ rows of the latter matrix according to a repetition pattern $(s_1, s_2 \ldots s_c)$ to make a $(k+c) \times (k+c)$ matrix $A_S$. The construction is shown in  \Cref{Fig:repetitions}. By an appropriate choice of $Y$, we can prove that $\Per A_S$ is proportional to $\Per X$. In particular, the result of \Cref{lem:rep_embedding} is obtained from Lemma 24 of \cite{grier_brod_2021} by replacing $B'_S$ and $B'$ with $A_S$ and $A$, $A$ with $X$, $k_s$ with $k$, and by setting $z=0$ and $y_{i,j}^{(l)}=1$. 

\begin{figure}
\centering    \includegraphics[width=0.75\columnwidth]{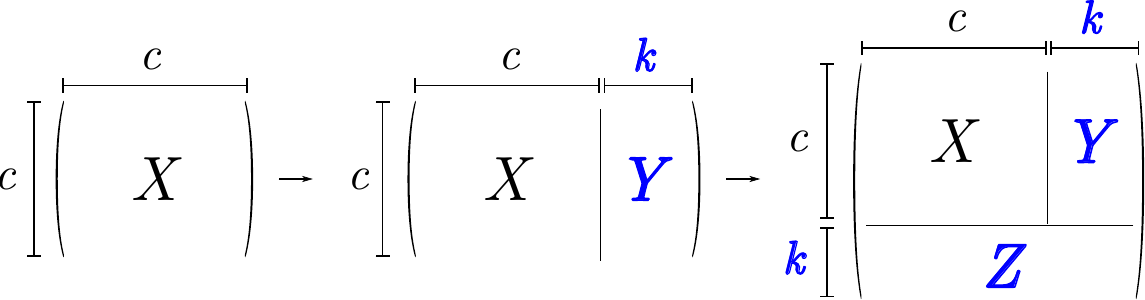}
    \caption{The figure shows how to extend the $c \times c$ matrix $X$ to handle $k$ row repetitions. The first step is to create the matrix $A=(X|Y)$ by appending $k$ extra columns, as described in the text. The matrix $Z$ just corresponds to all extra rows that are added according to the repetition pattern.}
    \label{Fig:repetitions}
\end{figure}

We illustrate the construction with a small, concrete example. Consider the following $3\times3$ matrix $X$:
\begin{equation}
X = \left(\begin{array}{c c c}
x_{1,1} & x_{1,2} & x_{1,3} \\
x_{2,1} & x_{2,2} & x_{2,3} \\
x_{3,1} & x_{3,2} & x_{3,3}
\end{array}\right),
\end{equation}
and suppose $S=(4,3,1)$, i.e., we need to repeat the first row three additional times, and the second row two additional times. This corresponds to an 8-photon BosonSampling outcome where we observed $c=3$ nonempty output modes. 

We now choose a $Y$ matrix as follows. For each \emph{extra} copy of a row we will need to add later, $Y$ must contain one unit-vector column with nonzero element matching that row. More concretely, if row $i$ will be repeated $s_i-1$ times, then $Y$ needs to have $s_i-1$ columns whose $i$th element is 1 and all others are zero. For our particular example we build the following matrix
\begin{equation} A= \label{eq:example}
\left(\begin{array}{ c c c | c c c c c}
x_{1,1} & x_{1,2} & x_{1,3} & 1 & 1 & 1 & 0 & 0 \\
x_{2,1} & x_{2,2} & x_{2,3} & 0 & 0 & 0 & 1 & 1\\
x_{3,1} & x_{3,2} & x_{3,3} & 0 & 0 & 0 & 0 & 0
\end{array}\right).
\end{equation}
And, after the repetition of all appropriate rows,
\begin{equation}
A_S = \left(\begin{array}{ c c c | c c c c c}
x_{1,1} & x_{1,2} & x_{1,3} & 1 & 1 & 1 & 0 & 0 \\
x_{1,1} & x_{1,2} & x_{1,3} & 1 & 1 & 1 & 0 & 0\\
x_{1,1} & x_{1,2} & x_{1,3} & 1 & 1 & 1 & 0 & 0\\
x_{1,1} & x_{1,2} & x_{1,3} & 1 & 1 & 1 & 0 & 0\\
x_{2,1} & x_{2,2} & x_{2,3} & 0 & 0 & 0 & 1 &  1\\
x_{2,1} & x_{2,2} & x_{2,3} & 0 & 0 & 0 & 1 &  1\\
x_{2,1} & x_{2,2} & x_{2,3} & 0 & 0 & 0 & 1 &  1\\
x_{3,1} & x_{3,2} & x_{3,3} & 0 & 0 & 0 & 0 & 0
\end{array}\right)
\end{equation}
It is easy to check that 
\begin{equation}
\Per(A_S) = 4!3!\Per(X).
\end{equation}
In the most general case, by using this choice for $Y$, we get that
\begin{equation}\label{eq:PerAsPoly}
\Per(A_S) = \Per(X) \prod_{i=1}^l s_i!
\end{equation}
where $k = n-c$. Note that, if $X$ is a $\{0,1\}$ matrix, so is $A_S$.

One way to see why this construction works is to consider the Laplace expansion for the permanent. Recursively apply the Laplace expansion $k$ times, eliminating the $k$ columns which do not include $X$, one at a time. It becomes clear that the choices of 0's and 1's in $Y$ are such that the only $c \times c$ submatrices of $A_S$ that contribute to the expansion are $(\prod_{i=1}^l s_i!)$ identical copies of $X$.

{We can now prove Theorem \ref{thm:worst}. }
\begin{proof}[Proof of Theorem \ref{thm:worst}]
    We begin by drawing a uniformly random outcome $S,$ i.e.\ $S\sim\mathcal U_{m,n}.$ By Lemma \ref{eq:largeDEVIATIONcliques}, {with probability $1-\exp(-\Theta(n))$} the outcome $S$ has $c_S=\Theta(n)$ clicks. Consequently, computing the permanent of matrices $X\in\{0,1\}^{c_S\times c_S}$ is $\sharP$-hard by the landmark result of \cite{valiant1979}. By \Cref{lem:rep_embedding}, taking as input our randomly drawn $S$ and a matrix $X\in\{0,1\}^{c_S\times c_S},$ in polynomial time one can construct our ``worst-case submatrix'' $A \in \{0, 1\}^{c_S \times n}$ in such a way that $\Per(A_S) = \Per(X) \prod_{i=1}^l s_i!.$ Thus $\Per(A_S)$ is $\sharP$-hard as well, completing the proof.  
\end{proof}

\subsection{Worst-to-average-case reduction}\label{sec:avgcase}

In this section we prove our main result, Theorem \ref{thm:main}, that in the saturated regime $\superap$ is $\sharP$-hard to additive accuracy $\exp(-O(n \log n))$. In Section \ref{sec:worst-case}, recall we proved the hardness of computing permanents for worst-case submatrices with average-case repetitions. Here, we present a worst-to-average-case reduction that proves $\sharP$-hardness for average-case submatrices with average-case repetitions---namely, the hardness of $\mathrm{SUPER},$ Problem \ref{prob:super}.
Were the robustness, meaning the additive error tolerance, of our result to be improved to $1/|\Smn| = \exp(-O(n)),$ it would prove quantum computational advantage conditional \emph{only} on the non-collapse of \textsf{PH}.
As mentioned in Section \ref{sec:intro}, the main challenge we need to overcome is correlations between matrix entries which breaks the original \cite{Aaronson2013,Bouland2021}-style proofs.

As a reminder, the overall scheme of these reductions is an interpolation argument inspired by Lipton's self-reducibility of the permanent, which exploits its polynomial structure to show that average-case instances are as hard as in the worst case \cite{lipton1991}. In particular, by taking a convex combination in variable $t$ of an average-case instance and a worst-case instance, the permanent is a univariate polynomial in $t$. 
Then, by estimating values of the polynomial for small $t$ by the average-case algorithm, one can extrapolate to $t=1,$ the permanents of which are proved hard in Section \ref{sec:worst-case}.

The interpolation is carried out by the Robust Berlekamp-Welch algorithm  developed in \cite{Bouland2021}, which adapts to polynomials over $\mathbb{C}$ the well-known Berlekamp-Welch algorithm \cite{BW} for polynomial interpolation over finite fields. For $m>2n^\gamma$, $\gamma\geq 1,$ Robust Berlekamp-Welch demonstrates that an algorithm to solve instances of Problem \ref{prob:super} to additive tolerance $\exp(-(\gamma + 4) n \log n - O(n))$ implies an efficient algorithm to compute permanents of worst-case-hard matrices under $\mathsf{BPP}^\mathsf{NP}$ reductions. Thus we demonstrate average-case hardness of $\superap$ to additive error $\exp(-5n\log n - O(n)).$ This robustness is slightly better than that known for $m=\Omega(n^2)$ BosonSampling---{with an exponent of} 5 rather than 6---however, on account of the smaller size of $\vert\mathcal S_{m,n}\vert$ in the saturated limit, the target robustness is commensurately larger, $\exp(-O(n))$. This is because the target additive error tolerance, or robustness, to which one seeks to prove hardness is on the order of a typical output probability, i.e.\ $1/|\Smn| = 1/\binom{m}{n} = \exp(-O(n)).$ Thus a small gap remains between the robustness we can prove and the robustness to which we conjecture hardness.  Closing the robustness gap remains open for all quantum advantage schemes.

The interpolation arguments of \cite{Aaronson2013} and \cite{Bouland2021} proceed by an entry-by-entry analysis of i.i.d.\ Gaussian matrices. As described in Section \ref{sec:intro}, for BosonSampling in the saturated limit the matrix entries instead come from a highly correlated measure. Although in the this setting this seemingly crucial i.i.d.\ property no longer holds, 
we nonetheless prove a total variation distance bound on the distribution of average-case instances under translation and dilation that allows one to recover the interpolation argument. Theorem \ref{thm:appx-invariance} shows that despite correlations, the distributions are perturbed only mildly under these transformations.

Our proof of Theorem \ref{thm:appx-invariance} develops tools that, given the ubiquity of the Haar measure over unitaries across the quantum information literature, we expect may be of independent interest. Starting with the probability density of the singular values of a submatrix of a Haar-random unitary, derived first by \cite{collins2003integrales}, we reduce the approximate invariance property to proving sharp tail bounds on the largest singular value. However, the bounds we need are stronger than those obtainable by blackbox concentration inequalities like Levy's lemma or log-Sobolev inequalities as used in prior work  \cite{singal2022implementation}. Instead, we carry out a more fine-grained analysis of the maximum singular value close to $1$ by directly analyzing the geometry of high-dimensional spheres. 

\subsubsection{Truncations of Haar-random unitaries: approximate invariance property}\label{ssec:appx-invariance}

As described above, in our worst-to-average-case reduction we smuggle a tiny amount of a worst-case matrix $A$ into a random matrix $B_t$ like so: 
\begin{equation}
B_t = (1-t) B_0 + t A \text{ for } t\in\left[0,1\right],
\end{equation}
in such a way that the permanent of $B_t$ is (by assumption in the reduction) approximately correct with high probability for small $t,$ but at $t=1$ is the permanent of {a matrix constructed from}  $A$, which is hard to compute. This is possible in the i.i.d.\  Gaussian setting of dilute BosonSampling because Gaussians are perturbed only slightly in TVD by shifting and rescaling in the manner above. It is not clear that this approximate shift-and-scale invariance should hold for the highly correlated measure on truncated Haar unitaries from which $B_0$ is drawn. Theorem \ref{thm:appx-invariance} shows that it does.

\begin{restatable}[Approximate shift-and-scale invariance]{thm}{appxinvariance}
\label{thm:appx-invariance}

Let $\mu$
be the probability measure on $n\times n$ complex-valued matrices induced by taking $n\times n$  submatrices of a Haar-random $m\times m$ unitary matrix, with each entry scaled by $\sqrt{m}$ to have unit variance. Fixing an arbitrary ``worst-case'' matrix $A\in\{0,1\}^{n \times n}$ and taking $B_0\sim \mu$, we define the random variable
\begin{equation}\label{eq:smuggled}
B_t = (1-t) B_0 + t A\ .
\end{equation}
We let $\mu_t$ be the probability measure for the random matrix $B_t$.
Note that $\mu_0=\mu$.

If $m = \lceil(2+\epsilon)n\rceil$ and $n\geq 2\epsilon^{-1}$ then there exists $C_\epsilon>0$ such that 
\begin{equation}
\|\mu_t - \mu\|_{TV} \leq C_\epsilon n^2 t\ .
\end{equation}
\end{restatable}

In other words, the distribution of $B_t$ is close to that of $B_0$ when $A$ is a $\{0,1\}$-valued matrix\footnote{More generally, it suffices for matrix $A$ to have bounded entries so that $\norm{A}_{\text{HS}} \leq \beta n$ for constant $\beta$.} and $t = O(1/n^2).$ 
{Note that matrix $B_t$ in Eq.\ (\ref{eq:smuggled}) is not yet a random matrix of the type that appears in $\superap$, which requires some of its rows to be repeated according to $S$. However, suppose that we erase the final $(n-c_S)$ rows of $B_t$ and replace them with some copies of its first $c_S$ rows. By the data processing inequality, the TV distance between distributions implied by Theorem \ref{thm:appx-invariance} is also an  upper bound on the TV distance between the distributions of the resulting matrices with repetitions. Therefore, the result of Theorem \ref{thm:appx-invariance} directly applies to the BosonSampling matrices in the saturated regime, corresponding to $c_S$ detector clicks and $n-c_S$ collisions.}

{For $m\leq 2n,$ one could alternately interpolate by the Cayley path  \cite{movassagh_quantum_2020,movassagh2023hardness} between worst and average-case instances {on the level of unitary group $\mathrm{U}(m)$}, as is done in Fermion Sampling \cite{fermionsampling}. However, the {final error} robustness is worse than that attained using current techniques because the degree of the resulting polynomial is $\Theta(n^2)$ (as opposed to $\Theta(n)$ using the approach approach from the Lipton-style proof for $m>2n$). We leave open whether or not it is possible to obtain similar robustness in the $m\leq 2n$ regime. }

We prove Theorem \ref{thm:appx-invariance} in Appendix \ref{sec:appx-invariance-pf}.

\subsubsection{Lipton-style interpolation and Robust Berlekamp-Welch} 

In Section \ref{sec:classical-hardness} we showed that Conjecture \ref{conj:UPEap}, namely that $\superap$ is $\sharP$-hard to additive error $\exp(-O(n))$, implies no efficient classical sampler. Here, we provide strong evidence in favor of Conjecture \ref{conj:UPEap}. In particular, in Theorem \ref{thm:main} we prove a necessary condition for the conjecture to hold. 

The idea of Lipton's proof of average-case hardness of the permanent is to take an entry-wise convex combination of a matrix $A_S$ whose permanent is hard and a matrix $B_0$ whose permanent we assume we can compute approximately correctly for most instances, producing matrix $B_t$ as in Eq.~(\ref{eq:smuggled}). The squared permanent of such an object is a univariate polynomial in $t$ of degree $2n.$ For small $t,$ points on this polynomial are {close to} instances of $\superap.$ At $t=1$ is a (squared) permanent proved $\sharP$-hard to compute in Section \ref{sec:worst-case}. Thus to show $\superap$ is $\sharP$-hard as well, one need only interpolate from points close to $t=0$ to $t=1.$ 
How close to $t=0$ our points must be so they are approximately correct with constant probability is established by Theorem \ref{thm:appx-invariance} as $t=O(1/n^2).$ {Note that Theorem \ref{thm:appx-invariance} is being invoked for $n \times n$  matrices $A_S$, $B_t$, and $B_0$ that do have repetitions according to pattern $S$, as they must correspond to either average instances of $\superap$ or its worst-case analogue (cf.\ {discussion below} Theorem \ref{thm:appx-invariance}).}

Robust Berlekamp-Welch, presented below as Theorem \ref{thm:robust_berlekamp_welch}, provides an algorithm for interpolation over the complex field. The algorithm runs in polynomial time with access to an $\textsf{NP}$ oracle, but this is sufficient as we are already proving $\sharP$-hardness under $\mathsf{BPP}^\mathsf{NP}$ reductions.
While one could alternately use standard Lagrange interpolation, this would suffice only to prove a weaker result with hardness to a smaller amount of additive error.

\begin{thm}
[Robust Berlekamp-Welch \cite{Bouland2021}]
\label{thm:robust_berlekamp_welch}
Let $p(t)$ be any degree $d$ polynomial. There is a $\mathsf{P}^\mathsf{NP}$ algorithm for the following task: \\
\begin{tabular}{r l}
\emph{Input:} & \begin{minipage}[t]{.8\linewidth}
List of data points $(t_1, y_1), \ldots, (t_{100d^2}, y_{100d^2})$ where the $t_i$ are equally spaced on the interval $[0, \Delta]$ ($\Delta < 1$) and {for every $i\in[100d^2]$} we have $\Pr[|y_i - p(t_i)| \ge \epsilon] < 1/4$.
\end{minipage} \\
\emph{Output:} & Estimate $z \in \mathbb C$ such that $|z - p(1)| \le \epsilon e^{d \log \Delta^{-1} + O(d)}$ with probability at least $3/4$.
\end{tabular}
\end{thm}

The error bound on the estimate to $p(1)$
arises because the maximum blowup under extrapolation from $\Delta$ to $1$ for a degree $d$ polynomial is $(1/\Delta)^d$.\footnote{Cf.\ Remez inequality.}
The subleading correction $\epsilon \exp(O(d))$ comes from the error incurred on $\left[0,\Delta\right]$ by points that are $\epsilon-$close to the true values, by an invocation of Paturi's bound \cite{Paturi1992}.

Our main result, stated informally in Section \ref{sec:intro} as Theorem \ref{thm:maininformal}, is below.
\begin{thm}[$\sharP$-hardness of $\mathrm{SUPER}$ in the saturated regime]\label{thm:main}
    Let $m\geq 2.1 n$. In the regime $m=\Theta(n)$ $\superap$ is $\sharP$-hard {under $\mathsf{BPP}^{\mathsf{NP}}$ reductions} to additive error {$\epsilon(S)=e^{-5n\log(n) - O(n)}$. with probability at least $1-\delta$, with $\delta=1/\poly{(n)}$}
\end{thm}
\begin{proof}
{Suppose we have an oracle $\mathcal{O}$ to solve $\superap$ with a failure probability $\delta$ (over the choices of $S$ and $U$). We will prove that we can therefore exactly compute the permanent of a specific worst-case $\{0,1\}$-valued matrix $X$ (as in Lemma \ref{lem:rep_embedding}) efficiently in $\mathsf{BPP}^{\mathsf{NP}^\mathcal{O}}$. Since exactly computing the permanent of a $\{0,1\}$-matrix is $\sharP$-hard, this will complete the theorem. Note that this is a $\mathsf{BPP}^{\mathsf{NP}}$ reduction, but this is not an issue since we are proving $\sharP$-hardness.}

{We begin by choosing a collision pattern $S$ at random and embedding matrix $X$ into a matrix $A_S$ with repetitions as in Lemma \ref{lem:rep_embedding}. Note that $S$ is compatible with this embedding with high probability, due to  Lemma \ref{lem:states-with-c-clicks}.
Our goal now is to use $\mathcal{O}$ to compute the squared permanents of {$B_t$, as per Eq.\ (\ref{eq:smuggled}), for $100d^2 = 400n^2$ values of $t$ in range $\left[0, \Delta\right]$}. Each matrix $B_t$ should have the same repetition pattern $S$ and correspond to the same worst-case matrix $A_S$.}

{By setting $\Delta=O(1/n^2)$ and invoking the data processing inequality on Theorem \ref{thm:appx-invariance}, we satisfy the requirements of \Cref{thm:robust_berlekamp_welch} that most points are $\epsilon$-close to the true values with sufficiently high probability. Given these values as input, by Theorem \ref{thm:robust_berlekamp_welch} the Robust Berlekamp-Welch algorithm outputs an approximation to $|\Per(A_S)|^2/m^n$ that is $\epsilon \; \exp(4n\log(n) + O(n))-$close with probability at least $3/4$. Equivalently, this corresponds to an approximation that is $\epsilon\ \exp(5n\log(n) + O(n))-$close to $|\Per(A_S)|^2$ for $m=O(n).$\footnote{The $m^n$ factor comes about because we rescale each row by $\sqrt{m}$ to apply Theorem \ref{thm:appx-invariance}, multiplying the permanent by $m^{n/2},$ then square the permanent.}}

{Now recall that $A_S$ is an integer matrix, therefore an estimate of its permanent with error less than $1$ is exact. Thus $\epsilon < \exp(-5n\log(n) - O(n))$ additive error on the squared permanents computed by $\mathcal{O}$ suffices to compute the exact permanent of $A_S$, which is $\sharP$-hard {by Theorem \ref{thm:worst}}.}

{The only outstanding issue is that this procedure chooses a collision pattern at random, whereas $\mathcal{O}$ only has to be precise with probability $1-\delta$ over choice of $S$ and $U$. It is consistent with this definition that there are some choices of $S$ for which $\mathcal{O}$ fails for almost every $U$ (as long as there are some other values of $S$ where it is unusually successful to compensate), and the conditions of Theorem \ref{thm:robust_berlekamp_welch} might not be satisfied, in which case there are no guarantees on the output of the Robust Berlekamp-Welch algorithm. However, note that the specific $S$ does not matter, since Theorem \ref{thm:worst} lets us embed $X$ into some $A_S$ for almost any $S$. Therefore, we can repeat this entire procedure several times and take the majority vote of the results to evade these pathological choices of $S$.}
\end{proof}


Notably, we achieve a robustness exponent of 5 compared to 6 for $m=\Omega(n^2)$ BosonSampling. {To see that, note that all of this procedure works when $m=\alpha n$ as long as, say, $\alpha \geq 2.1$. Carrying the dependence of $\alpha$ in the proof, we actually obtain a robustness requirement of $\epsilon \exp(4 n\log(n) + n\log(\alpha) + O(n))$. The term $n \log(\alpha)$ is $O(n)$ in the saturated limit, and $n\log (n)$ when $m=n^2$. Note, however, that the target robustness is also larger in the saturated limit in which $1/|\Smn| = e^{-O(n)}$, whereas when $m=n^2$ we have $1/|\Smn| = e^{-O(n \log(n))}.$} Closing the robustness gap altogether remains open for all quantum computational advantage proposals.

\section{Anticoncentration in the presence of correlations}\label{sec:anticonc}

In this section we comment on anticoncentration. Roughly speaking, anticoncentration is the property that most outcome probabilities for most BosonSampling experiments are not much smaller than the inverse of the domain size.
Although in existing quantum advantage arguments anticoncentration is formally neither necessary nor sufficient for hardness, this approximate ``flatness'' property nevertheless plays an important role. In particular, anticoncentration is a necessary ingredient to convert additive estimates of output probabilities, which we conjecture to be hard, to multiplicative estimates. In this way, it provides a ``sanity check'' on hardness conjectures that involve additive error. 
Moreover, classical easiness results often exploit failure to anticoncentrate \cite{harvardXEB, Napp2020}. While it remains open to prove anticoncentration for any variant of BosonSampling, there has been partial progress from e.g.\ \cite{nezami2021permanent} and strong numerical evidence in the $m=\Theta(n^2)$ regime. 

{In Section \ref{sec:intro} we discussed the two new challenges that arise when $m=\Theta(n)$ relative to the dilute regime, namely, correlations in the matrix entries and row repetitions. Though these might be expected to pose a threat to anticoncentration, below we present numerical evidence that it holds in this regime, and observe behavior similar to that in the i.i.d.\ Gaussian case}.

For $m=\Theta(n),$ the appropriate formulation of anticoncentration is as follows:

\begin{conj}[Anticoncentration of Probabilities]
\label{conj:anticobncentration_probs}
There exists some polynomial such that for all $n > 0$, $m \ge n$, $m=\Theta(n)$, and $\delta > 0$, 
\begin{equation}
\Pr_{\substack{U \sim \Haar(m) \\ S \sim\unifmn}} \left[ \frac{|\Per(U_S)|^2}{\prod_{i=1}^m s_i!} < \vert\Smn\vert^{-1} \frac{1}{\poly{m, 1/\delta}} \right] < \delta.
\end{equation}
\end{conj}

In Fig.~\ref{fig:anticoncentration}, we show numerical evidence for \Cref{conj:anticoncentration_probs} in the form of plots of random BosonSampling probabilities (i.e., $|\Per(U_S)|^2/\prod_i s_i!$). 
Recall that the conjecture posits that these probabilities are on the order of $1/\vert\Smn\vert$ and in particular do not grow smaller relative to this bound as $n$ grows. 
A key takeaway from the figure is that the difference between this bound and random probabilities is quite stable as $n$ increases (shown in particular on the right in Fig.~\ref{fig:anticoncentration}), consistent with the conjecture. If the trend we see numerically were to continue for larger permanents---we show data up to  30 photons---then this would imply the anticoncentration conjecture.

\begin{figure}
\centering
    \input{logProbabilityBoxPlot}
    \input{residualBoxPlot}
\caption{Box plots for the distribution of random probabilities in the $m = 2n$ regime. For each $n$, we calculated $\ln(|\Per(U_S)|^2/\prod_i s_i!)$ for 20  Haar random matrices $U \in \mathbb C^{2n \times 2n}$ and 20 uniformly random outcomes $S$. \textit{Left:} The blue box plots depict this distribution and the red dashed line show the anticoncentration bound (i.e., $- \ln \vert\Smn\vert$).
\textit{Right:} The blue box plots depict the same distribution shifted by the anticoncentration bound, hence why the red dashed line appears at $0$. Each box plot has the same format: min, 1st quartile, median, 3rd quartile, max.}
\label{fig:anticoncentration}
\end{figure}

Anticoncentration motivates a reformulation of SUPER in terms of multiplicative error.

\begin{conj}
\label{conj:UPEmp}The following family of problems, denoted $\supermp,$ is $\sharP$-hard: given $V\sim\Dmn$ as above, 
output $z \in \mathbb C$ such that 
\begin{equation}
\left|z - |\Per(V)|^2 \right| \le (\poly{m})^{-1} |\Per(V)|^2
\end{equation}

with probability at least $3/4$ over choice of $V\sim\Dmn$. 
\end{conj}

Conjectures \ref{conj:anticoncentration_probs} and \ref{conj:UPEmp} immediately yield that $\superap$ is $\sharP$-hard to error $  \prod_i s_i! \vert\Smn\vert^{-1},$ so we recover the same hardness condition as Conjecture \ref{conj:UPEap}, which by Section \ref{sec:classical-hardness} again implies no classical sampler.

\section{Gaussian BosonSampling} \label{sec:GBS}

In this section we show how our main results still apply to the BosonSampling variant known as Gaussian BosonSampling (GBS) \cite{hamilton2017gaussian}. The overall logic of the argument is the same, so we focus only on those steps of our main result that need to be modified in this case.

In Gaussian BosonSampling, rather than a Fock-state input, one uses a Gaussian state, such as a two-mode-squeezed state, as described in \cite{hamilton2017gaussian,deshpande2021GBS}. Usually, this means that the probabilities are written in terms of hafnians of particular transition matrices, rather than permanents. To sidestep having to worry about hafnians and to reuse much of our previous work, we instead use the construction of \cite{grier_brod_2021}, known as BipartiteGBS. 

\begin{figure}[t]
\centering
    \includegraphics{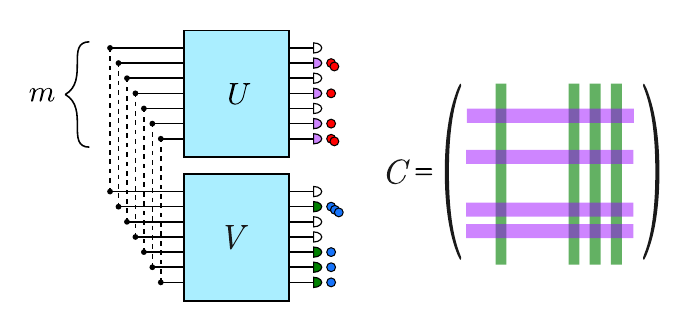}
\caption{Gaussian Boson Sampling. Vertical dashed lines correspond to preparation of two-mode squeezed states with identical squeezing parameters $r$. The overall transition matrix is given by $C = (\tanh{r}) U W^\dag$, and transition probabilities are given by permanents of submatrices of $C$, as described in the main text. Each photon observed in one of the top (bottom) $m$ modes selects one of the rows (columns) of $C$.}
\label{fig:setupGBS}
\end{figure}

We initialize each of $2m$ modes in the same Gaussian state, $\exp\left[r ({a^\dag}^2-a^2)\right] \ket{0}$, where $r$ is known as the squeezing parameter. Then this state evolves in a 2$m$-mode inteferometer as shown in Fig.\ \ref{fig:setupGBS}. Let us write an output state as $\ket{S,T} = \ket{s_1, s_2, \ldots, s_m, t_1, t_2, \ldots, t_m}$ where, due to a well-known property of two-mode squeezed vacuum states, $\sum_i s_i = \sum_i t_i =: n$\footnote{Here, we consider $m$ and $n$ to be the number of modes and photons on either half of the interferometer, rather than in total. This is convenient so that the permanents we consider are $n \times n$ as in the previous case.}. The probability of observing outcome $\ket{S,T}$ is then given by
\begin{equation}
\label{Eq: main_dbn}
q(S,T) = \frac{1}{\normZ} \frac{\left\lvert \Per(C_{S,T})\right\lvert ^2}{\prod_i s_i! \prod_j t_j!}\ .
\end{equation}
In \cref{Eq: main_dbn}, $C = (\tanh{r}) U V^\dag$ is a transition matrix, where $U$ and $V$ correspond to two $m$-mode interferometers and $C_{S,T}$ is a submatrix constructed from $C$. This submatrix follows exactly the same prescription for choosing (and repeating) rows and the standard BosonSampling one described in Section  \ref{sec:sketch}, but now we select the $i$th row $t_i$ times and the $j$ column $s_j$ times. Finally, $\normZ = (\cosh r)^{2m}$ is a normalization.

In our case, {we choose $U$  to be Haar-random and and $V$ to be identity matrix}. Under these choices, the resulting GBS model of Eq.\ (\ref{Eq: main_dbn}) differs from the BosonSampling model we considered in Sec.\ \ref{sec:sketch} only in a few  aspects: first, the submatrices of $C$ (which is itself proportional to a Haar-random unitary matrix) are built by repeating both rows and columns; second, the number of particles is not fixed; and finally, the probabilities are not given only by permanents, but are multiplied by the $1/\normZ$ factor.

Let us first address the issue of variable photon number. Since we are assuming all squeezing parameters are the same, the number of photons emitted in each of the $m$ top input modes independently follows a geometric with mean $\left \langle n_i \right \rangle=\sinh(r)^2$ . Consequently, the total number of photons $n$ in the top $m$ modes follows negative binomial distribution $P(n)= \binom{n+m-1}{n}\frac{\tanh(r)^{2n}}{\cosh(r)^{2m}}$ with mean $\left \langle n \right \rangle = m \sinh(r)^2$ and variance $\mathrm{Var}(n)= m \frac{\cosh(r)^2}{\sinh(r)^4}$ \cite{Jex2019}. The mode (most likely result) negative binomial distribution $P(n)$ equals $n_\ast=\lfloor (m-1) \sinh(r)^2 \rfloor$.   Since $r$ can be treated as a constant we get that $P(n_\ast)\geq\frac{C}{\sqrt{m}}$, where $C$ is a constant depending only on $r$.
This means that we can in principle postselect on observing exactly $n_\ast$ photons on each half of the experiment by repeating the experiment only a polynomial number of times. Therefore, for simplicity, from now on we assume that we are working only in the subspace with \emph{exactly} $n_\ast = \lfloor m \sinh(r)^2 \rfloor$ photons. In particular our result holds in the regime where $m>2.1 n$, which corresponds to a squeezing parameters of roughly $r = 0.65$, well within experimentally-accessible values \cite{zhong2021phase}.

{The postselected outcome probability corresponding to $m$ modes and $n$ input photons in each half of the interferometer then equals\footnote{This follows from the observation that $q(S,T)=f(n) p_U(S,T)  $, where $f(n)$ is as some function of the photon number $n$ and $p_U(S,T) =\frac{\left\lvert \Per(U_{S,T})\right\lvert ^2}{\prod_i s_i! \prod_j t_j!}$  can be regarded as a probability in $S\in \Smn$ (for fixed $T\in\Smn$) or vice versa. }}

\begin{equation}
\label{Eq: main_mo}
p_U(S,T) = \frac{1}{|\Smn|} \frac{\left\lvert \Per(U_{S,T})\right\lvert ^2}{\prod_i s_i! \prod_j t_j!}.
\end{equation}

We now sketch how the proof steps of previous sections would be adapted to the case of GBS.

\subsection{Worst-case hardness of permanents with repetitions revisited}

Recall that, in Sec.\ \ref{sec:worst-case}, we showed $\sharP$-hardness of worst-case permanents when rows were repeated. More concretely, we started with some $c \times c$ matrix $X$, we extended it into a $c \times (c+k)$ matrix $A$ by appending columns according to a specific prescription, and then we repeated its rows according to some collision pattern to construct matrix $A_S$. The result was Lemma \ref{lem:rep_embedding}, which showed that $\Per A_s = \Per X \prod s_i! $. The proof of that theorem followed directly from Lemma 24 in \cite{grier_brod_2021} by a particular choice of all relevant variables.

Lemma 24 in \cite{grier_brod_2021} is, however, used to prove a stronger result, namely Theorem 12 therein, which encompasses repetitions of both rows and columns. For our case, this simply reduces to applying the construction of Fig.\ \ref{Fig:repetitions} twice, once for rows and once for columns. 

More concretely, suppose we have two collision patterns in the output, $T$, for the top $m$ modes and which selects columns, and $S$, for the bottom $m$ modes and which selects rows. Let $k_T$ and $k_S$ be the (possibly distinct) number of \emph{collisions} on each respective pattern. By replacing $A$ with $X$, $z=0$ and $x_{i,j}^{(l)}=y_{i,j}^{(l)}=1$, the statement of Theorem 12 of \cite{grier_brod_2021} reduces to
\begin{lem}
\label{lem:rep_embedding_gbs}
Let $m\geq n$ and $S,T\in\mathcal{S}_{m,n}$ be two collision patterns such that (i) $S$ and $T$ have $k_S$ and $k_T$ collisions respectively, (ii) the total photon number is $n = l+k_S+k_T$, and (iii) $S$ and $T$ have at least $k_T$ and $k_S$ entries with a single photon each, respectively.

Given a matrix $X \in \{0, 1\}^{l \times l}$, there is a poly-time constructible matrix $A \in \{0, 1\}^{(l+k_T) \times (l+k_S)}$ such that
\begin{equation}
\Per(A_{S,T}) = \Per(X) \prod_{i=1}^l s_i!t_i!.
\end{equation}
Note that condition (iv) is guaranteed as long as $l\geq k_T,k_S$.
\end{lem}
Note that the fact that we need to repeat both rows and columns adds a small overhead to the size of these matrices. In particular, whereas in the case of BosonSampling we started with a $c \times c$ matrix $X$ where $c$ was the number of detectors clicks, now we construct $(l+k_T) \times (l+k_S)$ matrix whose rows and columns are then repeated, implying that we have $l+k_S$ clicks on the top side and $l+k_T$ clicks on the bottom, for a total of $n:=l+k_S+k_T$ photons on each side. Furthermore, the explicit construction requires there to be $k_S$ non-collision output modes on the $T$ side and $k_T$ non-collision modes on the $S$ side.
Nonetheless, from these arguments it follows that permanents with repetitions in both rows and columns are also $\sharP$ to compute in the worst case.

\subsection{Combinatorics}

We now need to prove that the analogous of Lemma \ref{eq:largeDEVIATIONcliques} in the case of GBS, which would imply that \emph{most} outcomes are hard according to Lemma \ref{lem:rep_embedding_gbs}, in the dilute regime. To that end, first note that if we trace out  either the top or the bottom half of the modes in the setup of Fig.\ \ref{fig:setupGBS}, the state of the remaining modes is the maximally mixed state of $n$ photons in $m$ modes. 
This means that each half of the outcomes follows, independently, the combinatorics of Lemma \ref{eq:largeDEVIATIONcliques}.

One additional issue in the GBS case is that it is not enough to guarantee that there are many detector clicks on either side. As discussed just after Lemma \ref{lem:rep_embedding_gbs}, we also need the output state to have $k_S (k_T)$ modes with a single particle each on the top (bottom) half. This is likely an artifact of the proof technique of Theorem 12 in \cite{grier_brod_2021}, but we currently do not have an alternative proof of Lemma \ref{lem:rep_embedding_gbs} without this requirement. Nonetheless, the following Lemma guarantees that most outcomes types encode a sufficiently hard problem as long as {$m\geq 2.1n$} \footnote{{The condition $m\geq 2.1n$ is not fundamental for our combinatorial considerations and we decided to impose it for the sake of simplicity. However, this bound is also required by the worst-to-average case reduction that underpins Theorem \ref{thm:mainG} (since it also utilizes random matrix theory from Theorem \ref{thm:appx-invariance}).}}:
\begin{lem} 
\label{lem:combinatoricsGBS}
{Let $m\geq 2.1 n$ and let $S,T\in \mathcal{S}_{m,n}$ be collision patterns associated with the bottom and top half of our GBS setup chosen independently from $\unifmn$. Let $k_S$, $k_T$ denote the number of colliding photons in  $S$ and $T$ respectively. Let $nc_S,nc_T$ denote the number of modes that are occupied and exhibit no collisions in $S,T$, and let $c_S, c_T$ be the number of clicks in $S,T$, respectively. Then, with probability greater than $1- \exp(-\Theta(n))$ we have
\begin{align}
    k_T \leq n/3 \leq nc_S\ \label{eq:fb} \\
    k_S \leq n/3 \leq nc_T\ \label{eq:sb}  .
\end{align}
Consequently, typical pairs of outputs $S,T\in \mathcal{S}_{m,n}$ can be used to compute permanents of $X \in \{0, 1\}^{l \times l}$ for $l\geq n/3 $ (in the sense of Lemma \ref{lem:rep_embedding_gbs}).
}
\end{lem}
{
\begin{proof}
We first note that we have the bound $nc_S \geq c_S - k_S$ (and similarly for output $T$), 
    which follows from the fact that no collision clicks (modes) in $S$ are the ones that are (i) occupied by a single particle and (ii) are not occupied by any of the remaining $k_S$ ``excess'' particles. Using $k_S=n-c_S$ (and the same relation for outcome $T$) we get 
\begin{equation}
     nc_S \geq 2c_S - n\ ,\ \ nc_T \geq 2c_T - n\ .
\end{equation}
Using Lemma \ref{lem:manyclicks} and recalling the notation $\alpha=m/n$ we get that for any $t\geq 0$
\begin{equation}\label{eq:ncbound}
    \Prob_{S\sim\unifmn}\left[ nc_S \leq\frac{\alpha -1+1/n}{\alpha+1+1/n} n - 2 t n    \right]\leq 2\ \exp(-2 t^2 n)\ .
\end{equation}
Additionally from $k_S=n-c_S$ we get
\begin{equation}\label{eq:ksbound}
    \Prob_{S\sim\unifmn}\left[ k_S \geq  \frac{1+1/n}{\alpha+1+1/n} n + t n    \right]\leq 2\ \exp(-2 t^2 n)\\ . 
\end{equation}
We have identical bounds for $nc_T,k_T$, where $T\sim\unifmn$. Using the above bounds in conjunction with the union bound it is now easy to see that, for $\alpha\geq 2.1$ and $t$ set to a small enough constant (independent of $n$ and $m$), we get that with probability at least $1- \exp(-\Theta(n))$ (over the choice of $S,T \sim \unifmn$) we get inequalities \eqref{eq:fb} and \eqref{eq:sb}.
 \end{proof}
}

\subsection{Gaussian BosonSampling in the saturated limit}

The big picture argument for the hardness of GBS in the saturated limit follows very closely the argument for standard BosonSampling made in the rest of the paper. We begin by positing a new problem tailored for GBS:

\begin{problem}[Gaussian Sub-Unitary Permanent Estimation with Repetitions, $\gsuperap$]
\label{prob:gsuper}
We define a family of problems parametrized by error scaling function $\varepsilon \colon \mathbb N^m \to \mathbb R^+$ and failure probability $\delta$. Let $\Gmn$ denote a probability distribution on $n \times n$ matrices $V$ obtained by sampling $\{S,T\}\sim\unifmn$ and $U\sim\Haar(m)$, and setting $V=U_{S,T}$. Given $V\sim\Gmn$ and error probability $\delta>0$, output $z\in\mathbb{C}$ such that $|z - |\Per(V)|^2 | \le \varepsilon(S,T)$
with probability at least $1 - \delta$ over choice $V$\footnote{Again, $\varepsilon(S,T)$ is well-defined as it can only depend on the multiplicities of rows and columns of $V$.}.
\end{problem}

The analogue of Conjecture \ref{conj:UPEap} for $\gsuperap$ is

\begin{conj}\label{conj:GUPEap}
$\gsuperap$ is $\sharP$-hard to additive imprecision $\varepsilon(S,T)=\frac{1}{\poly{n}}\prod_{i,j} s_i! t_j! /\vert\Smn\vert$, when ${m\geq 2.1 n}$ with failure probability $\delta\leq\frac{1}{4}$.
\end{conj}

We now need to prove analogs of the main two results from previously in the paper, namely Theorems \ref{thm:no-sampler}
and \ref{thm:main}. Since much of the argument is similar, we will focus on the differences that appear in the case of GBS.

Let us begin with the analog of Theorem \ref{thm:no-sampler}. Recall that the gist of this Theorem, which is proven in Appendix \ref{sec:stockmeyer}, is that, if one has an efficient classical sampler capable of simulating a BosonSampling instance to within small total variation distance in the saturated regime, then this classical sampler can be input to Stockmeyer's algorithm and produce a good approximation for the individual probabilities of the corresponding BosonSampling instance. This is a very generic argument, used for many quantum advantage proposals, and does not depend on any particularity of BosonSampling. The main difference that arises is that the output sample space now has dimension $\vert\Smn\vert^2$ rather than $\vert\Smn\vert$, since the GBS outputs consist of pairs of strings $\{S,T\}$ of $n$ photons in $m$ modes each, but this is canceled by the multiplicative factor in Eq.~(\ref{Eq: main_mo}) due to how the permanent function is encoded in a GBS amplitude. Keeping track of both these changes, we obtain:

\begin{thm}\label{thm:no-samplerGBS}
    Let $m=\Theta(n), m\geq 2.1 n$. An existence of a classical sampler that samples from a distribution approximating the GBS distribution (i.e., Eq.~(\ref{Eq: main_mo})) to accuracy $\beta=1/\poly{n}$ implies a solution to $\superap$ with additive imprecision ${\epsilon(S,T)}=\frac{1}{\poly{n}}\prod_{i,j} s_i! t_j!/\vert\Smn\vert$ in $\mathsf{BPP}^\mathsf{NP}. $ Assuming Conjecture \ref{conj:GUPEap}, this collapses $\mathsf{PH}$. 
\end{thm}

We now move to the analog of Theorem \ref{thm:main}. Recall that, at a high level, the proof of Theorem \ref{thm:main} proceeds in the following steps: (i) we prove that most outcomes of a BosonSampling experiment encode a $\sharP$-hard permanent, via Theorem \ref{thm:worst} and Lemmas \ref{lem:states-with-c-clicks} and \ref{lem:rep_embedding}; (ii) we prove that it is possible to smuggle a small amount of a hard matrix into a random matrix from the relevant ensemble (i.e., Theorem \ref{thm:appx-invariance}), and that then the permanent of the hard matrix retrieved from several estimates of the permanent of random matrices via polynomial interpolation (i.e., Theorem \ref{thm:robust_berlekamp_welch}); finally, (iii) the proof of Theorem \ref{thm:main} concatenates the previous two steps with the details of the definition of $\superap$ to obtain the main result. 

In the case of $\gsuperap$, step (i) is replaced by a combination of Lemmas \ref{lem:rep_embedding_gbs} and \ref{lem:combinatoricsGBS}. These together show that, with high probability, the collision pattern of outcomes of a typical GBS experiment are compatible with encoding a $\sharP$-hard problem. Step (ii) is actually unchanged, since Theorem \ref{thm:appx-invariance} works just as well for repetitions of both rows and columns as it did for rows, and Theorem \ref{thm:robust_berlekamp_welch} just depends on the polynomial nature of the permanent. Finally, the logic of step (iii) also remains the same, though we need to keep track of numerical factors. Going through all of these in detail, we arrive at the following

\begin{thm}[$\sharP$-hardness of $\mathrm{GSUPER}$]\label{thm:mainG}
    In the regime where $m=O(n)$, $\gsuperap$ is $\sharP$-hard to additive error {$\epsilon(S)=e^{-5n\log(n) - O(n)}$, with probability at least $1-\delta$, with $\delta=1/\poly{(n)}$}.
\end{thm}

\section{Discussion} \label{sec:conclusions}

Our work better connects the theory of BosonSampling to its experimental implementation.
Aaronson and Arkhipov's foundational work led to a number of important extensions which improved our understanding of the complexity of BosonSampling---from generalizations to Gaussian BosonSampling, to improving the robustness of average-case hardness arguments, to efficiently spoofing or verification of experiments, to characterizing the effects of noise. 
While many of the more empirical works have focused on the saturated regime due to its connection with experiment, most of the theoretical arguments have  focused on dilute regime and may need to be re-investigated in this new context.

Since the original version of this manuscript was released, a subsequent paper of Bouland, Datta, Fefferman, and Hern\'{a}ndez \cite{bouland2024averagecasehardnessbosonsampling} improved the robustness of hardness for dilute BosonSampling to $e^{-n\log n - n - O(n^\delta)}$ for any $\delta>0$, i.e.\ within a $O(n^\delta)$ in the exponent of the desired robustness.
This technique immediately applies to our results, and yields robustness for Fock basis BosonSampling in the saturated regime within a factor of $O(n^\delta)$ in the exponent of the target hardness as well.
However, we note other aspects of their work do not immediately carry over---such as their hardness of average-case exact sampling  results---as these depend on specific properties of the Gaussian measure which do not hold in the saturated regime. We leave this as an open problem.

Many interesting questions remain.
Do spoofing algorithms for BosonSampling become easier or harder in the saturated case?
How few modes are needed for intractability, for example, for which $\alpha$ do $m=\alpha n$ experiments have the best evidence for hardness,
and are there any fundamental limits on this constant? 

\section*{Acknowledgments}
A.B. and I.D. were supported in part by the AFOSR under grants FA955024-1-0089 and FA9550-21-1-0392.
A.B. and B.F. were supported in part by the DOE QuantISED grant DE-SC0020360. 
A.B. was supported in part by the U.S. DOE Office of Science under Award Number DE-SC0020266.
D.J.B. acknowledges financial support from Brazilian funding agencies FAPERJ, CNPq, and Instituto Nacional de
Ciencia e Tecnologia de Informa\c{c}\~ao Qu\^antica (INCT-IQ).
B.F. acknowledges support from AFOSR (FA9550-21-1-0008). M.O was  supported by the TEAM-NET project
M.O acknowledges the support from the European
Union’s Horizon Europe research and innovation program under EPIQUE Project GA No. 101135288.  This material is based upon work partially
supported by the National Science Foundation under Grant CCF-2044923 (CAREER) and by the U.S. Department of Energy, Office of Science, National Quantum Information Science Research Centers (Q-NEXT).
F.H. was supported by the Fannie and John Hertz Fellowship.
This work was done in part while A.B., I.D., B.F., D.G. and M.O. were visiting the Simons Institute for the Theory of Computing.

\bibliographystyle{alpha}
\bibliography{bosref.bib}

\newpage
\section*{Appendices}

\appendix

\section{Stockmeyer reduction for hardness of sampling}\label{sec:stockmeyer}

In this section we prove Theorem \ref{thm:no-sampler}, which shows that existence of   approximate classical sampler for BosonSampling implies the solution of $\superap$ problem to precision inversely proportional to $|\Smn|$, the dimension of symmetric subspace.

We first formalize the notion of a classical sampler

\begin{defn}[Classical sampler]\label{def:sampler}
A classical sampler is a probabilistic polynomial-time algorithm that, given as input $m\times m$ matrix $U$, and error $\beta\geq0$ 
, outputs a sample $S=(s_1,\ldots,s_m)$ from distribution $q_U$ such that 
\begin{equation}
\norm{q_U-p_U}_{TV} \leq \beta\ ,
\end{equation}
where $\mathcal{P}_U$ is the outcome distribution of the BosonSampling experiment.
\end{defn}

As is standard among hardness of sampling results, we reduce classical sampling to approximately computing output probabilities. The reduction uses Stockmeyer's approximate counting algorithm \cite{stockmeyer1983complexity}, which runs in \textsf{BPP} with access to an \textsf{NP} oracle\footnote{The idea of Stockmeyer's algorithm is to estimate the probability of any outcome by estimating the number of random strings that cause the sampler to output that outcome. This uses that a classical randomized algorithm can be treated as a deterministic algorithm that takes a random input.}. It is the ability in $\mathsf{BPP}^\mathsf{NP}$ to approximate output probabilities, which are expressed in terms of squared matrix permanents we conjecture to be $\sharP$-hard, that draws a contradiction if \textsf{PH} is infinite.

To prove the Stockmeyer reduction we require a technical lemma we term ``reverse embedding,'' which allows us to formally connect the input to $\superap$, i.e.\ Problem \ref{prob:super}, which is an $n\times n$ matrix drawn from the correlated distribution $\Dmn$ of sub-unitaries with row repetitions, and Stockmeyer's algorithm, which receives as input a unitary matrix via the classical sampler of Definition \ref{def:sampler}. Reverse embedding arises whenever there is a mismatch between algorithm input types. Notably, this technicality is absent in Random Circuit Sampling, where the input to both the classical sampler and Stockmeyer is the circuit unitary.

\begin{lem}[Reverse Embedding]\label{lem:reverse-embedding}
    Given as input an $n\times n$ matrix $V$ drawn from distribution $\Dmn$ as in Problem \ref{prob:super}, in $\mathsf{BPP}^\mathsf{NP}$ one may output an $m\times m$ matrix $U=\mathcal{A}(V)$ which is distributed according to Haar measure, as well as some $S$ such that $U_S= V$. 
     
\end{lem}
In particular, the distinct rows of $U_S$ form a submatrix of the $m \times n$ truncation of $U,$ i.e.\ its leftmost $n$ columns.

\begin{proof}
    In \textsf{BPP}, generate unitary $U\sim\Haar(m)$ and $S \sim \unifmn$, and with the \textsf{NP} oracle, post-select on the submatrix $U_S$ being identical to $V$. By the definition of a marginal distribution, such a $U$ and $S$ exist.
\end{proof}
{
\begin{thm}[Stockmeyer reduction for BosonSampling] \label{thm:stockmeyer}
If there exists a classical sampler as in Definition \ref{def:sampler}, then there exists in $\mathsf{BPP}^\mathsf{NP}$ a solution to $\superap$ to additive imprecision 
\begin{equation}\label{eq:addERROR}
 \epsilon(S)=\frac{\prod_{i=1}^m s_i!}{|\Smn|\delta}\left(4 \beta + \frac{1}{\poly{m,n}}\right)\ ,   
\end{equation}
 with probability at least $1-\delta$. The numbers $(s_i)$ come from $S=(s_1,\ldots,s_m)$. Furthermore, assuming that $\beta=1/\poly{m,n}$, $\delta=\Theta(1)$ we get that there exist a $\mathsf{BPP}^\mathsf{NP}$ solution solution to $\superap$ with additive imprecision
 \begin{equation}
     \epsilon(S)=\frac{\prod_{i=1}^m s_i!}{|\Smn|\poly{m,n}}\ ,
 \end{equation}
with failure probability at most $\delta=\Theta(1)$.
\end{thm}
}

\begin{proof}

{
 Let us fix $U\in\mathrm{U}(m)$ and observe that because $p_U$ and $q_U$ are $\beta-$close in total variation distance, $p_U(S)$ and $q_U(S)$ are close with high probability for randomly chosen $\S\sim\unifmn$:
\begin{equation}
    \mathbb E_{S\sim \unifmn} \vert p_U(S)-q_U(S)\vert = 
\frac{1}{\vert\Smn \vert} \sum_{S \in \Smn} \vert p_U(S)-q_U(S) \vert  \leq
\frac{2\beta}{\vert\Smn\vert}.
\end{equation}
By Markov's inequality, for any $k_1>0$ 
\begin{equation}\label{eq:EQ1}
\Pr_{S\sim \unifmn}\left[ \vert p_U(S)-q_U(S) \vert \geq \frac{2\beta k_1}{\vert \Smn\vert} \right] \leq \frac{1}{k_1}\ .
\end{equation}
In other words, the output of the classical sampler must be accurate for most outcomes $S$. Now, in order to obtain a good approximation to $q_U(S)$ and therefore by proxy $p_U(S)$ by the above, we run Stockmeyer's algorithm that in $\mathsf{BPP}^\mathsf{NP}$ takes as input unitary a $U$ and for any outcome $S$ estimates $q_U(S)$ as $\widetilde{q}_U(S)$ such that 
\begin{equation}\label{eq:STOCK}
 |\widetilde{q}_U(S) - q_U(S)| \leq \frac{1}{\poly{m,n}}\cdot q_U(S)\ .
\end{equation}
The accuracy of estimate $\widetilde{q}_U(S)$ depends on the magnitude of $q_U(S),$ which trivially satisfies 
$\mathbb E_{S\sim\unifmn} [q_U(S)] \leq |\Smn|^{-1},$ as $q_S\leq 1.$ So again by Markov,  for any $k_2>0$ 
\begin{equation}\label{eq:EQ2}
\Pr_{S\sim\unifmn} \left[q_U(S) \geq \frac{k_2}{|\Smn|}\right] \leq \frac{1}{k_2}.
\end{equation}
By combining \eqref{eq:EQ1} and \eqref{eq:EQ2} and employing \eqref{eq:STOCK} we get
\begin{equation}
    \Pr_{S\sim\unifmn}\left[ \vert p_U(S)-  \widetilde{q}_U(S) \vert \geq \frac{1}{|\Smn|}\left(2 \beta k_1 + \frac{k_2}{\poly{m,n}}\right) \right]\leq \frac{1}{k_1}+ \frac{1}{k_2}\ .
\end{equation}
By setting $k_1,k_2=2/\delta$ we get.
\begin{equation}
    \Pr_{S\sim\unifmn}\left[ \vert p_U(S)-  \widetilde{q}_U(S) \vert \geq \frac{1}{|\Smn|\delta}\left(4 \beta + \frac{1}{\poly{m,n}}\right) \right]\leq \delta\ .
\end{equation}
Since the above inequality hold for arbitrary fixed $U$ we can also write 
\begin{equation}\label{eq:intermetiate}
    \Pr_{\substack{U \sim \Haar(m) \\ S \sim\unifmn}}\left[ \vert p_U(S)-  \widetilde{q}_U(S) \vert \geq \frac{1}{|\Smn|\delta}\left(4 \beta + \frac{1}{\poly{m,n}}\right) \right]\leq \delta\ ,
\end{equation}
where the distributions over $U$ and $S$ are independent. 
Now, in order to connect the above estimate to the definition of $\superap$ we note that (i) for randomly chosen pairs $(U,S)$ as in \eqref{eq:intermetiate} we have $p_U(S)=|\Per(V)|^2/\prod_{i\in[m]} s_i!$ and furthermore (ii) for $V\sim\Dmn $  we can in  $\mathsf{BPP}^\mathsf{NP}$ (c.f. Lemma \ref{lem:reverse-embedding}) output $U=\mathcal{A}(V)$ that is Haar distributed as well as $S$ that is distributed according to $\unifmn$. Consequently it follows that
\begin{equation}\label{eq:intermetiate}
    \Pr_{V\sim\Dmn}\left[ \left\vert \frac{|\Per(V)|^2}{\prod_{i=1}^m s_i!}-  \widetilde{q}_{\mathcal{A}(V)}(S) \right\vert \geq \frac{1}{|\Smn|\delta}\left(4 \beta + \frac{1}{\poly{m,n}}\right) \right]\leq \delta\ .
\end{equation}
We conclude the proof by noting that\footnote{Note that the output $S$ can be easily inferred from $V$ and thus can be regarded as a (poly-computable) function of the input $V\sim\Dmn$.} $ (\prod_{i=1}^m s_i!) \widetilde{q}_{\mathcal{A}(V)}(S)$ is the desired approximation to $|\Per(V)|^2$ obtained in $\mathsf{BPP}^\mathsf{NP}$. 
}

\end{proof}

\newpage
\section{Proof of shift-and-scale invariance}\label{sec:appx-invariance-pf}
\
\appxinvariance*

As a first step, we note that we can express the probability measure $\mu_\theta$
using dilations $D_\rho$ and
translations $T_H$ applied to the measure $\mu$.
We will define these operations in terms of their actions on probability densities;
we write $\mu = P(B)dB$ where $P\in L^1(\Complex^{n\times n})$ is a continuous density.

Given $\rho\in\Complex$, the dilation
operator $D_\rho$ is defined on probability densities on $\Complex^{n\times n}$ by
\begin{equation}
    (D_\rho P)(B) = |\rho|^{-2n^2} P(B/\rho)\ .
\end{equation}
The factor of $\rho^{-2n^2}$ comes from the fact that the (real) dimension of
$\Complex^{n\times n}$ is $2n^2$, and ensures that $D_\rho P$ remains a probability density. 

The translation $T_H$ is defined by
\begin{equation}
    (T_H P)(B) = P(B-H)\ .
\end{equation}

It is clear that $T_A P$ is the probability density of the distribution of $B+A$
when $B$ is sampled from $P(B)dB$.

The proof of Theorem \ref{thm:appx-invariance} is divided into two steps, given as Lemmas \ref{lem:approximate-invariance} and \ref{lem:gradient-bounds}. 
With the definitions above, we can say that if $B_0$ is sampled from the probability
density $P$, then $B_t$ is sampled from $T_{t A}D_{1-t} P$.
Lemma \ref{lem:approximate-invariance} characterizes the total variation distance between these
probability distributions in terms of an integral of $|\nabla P|$, for which we prove a bound in Lemma \ref{lem:gradient-bounds}.

For convenience, given matrices $A,B\in\Complex^{n\times n}$
we write $|A|$ for the Hilbert-Schmidt norm 
and $A\cdot B$ for the corresponding inner product. Below, it is often convenient to use $N\coloneqq n^2.$

\begin{lem}
\label{lem:approximate-invariance}
Let $P\in L^1(\Complex^N)$ be any integrable function.  Then for any
$A$,
\begin{equation}
\label{eq:translation-invariance}
\int |P(B) - T_A P(B)| \diff B \leq \int |A\cdot\nabla P(B)|\diff B.
\end{equation}
Moreover for any $|\delta|\leq \frac{1}{10}$,
\begin{equation}
\label{eq:dilation-invariance}
\int |P(B) - D_{1+\delta}P(B)| \diff B
\leq 10 |\delta| \big(N\int |P(B)|\diff B + \int |B| |\nabla P(B)|\diff B).
\end{equation}
\end{lem}
\begin{proof}
First we prove~\eqref{eq:translation-invariance} under the assumption that $P$
is continuously differentiable.  In this case we use the fundamental theorem of calculus 
in the form $f(1)-f(0) = \int_0^1 f'(t)\diff t$ applied
to the function $f(t) = P(B-tA)$, which gives
\begin{equation}
P(B) - P(B-A) = -\int_0^1 A \cdot \nabla P(B-tA) \diff t.
\end{equation}
Writing $P(B-A) = T_AP(B)$ and using the triangle inequality above we estimate
\begin{align*}
|P(B)-T_AP(B)| &= |P(B)-P(B-A)| \\
&\leq \int_0^1 |A\cdot \nabla P(B-tA)| \diff t.
\end{align*}

Integrating over $B$ and applying Fubini's theorem and then using the change of variables $Y=B-tA$
we have
\begin{equation}
\begin{split}
\int |P(B) - T_AP(B)| \diff B
&= \int \int_0^1 |A\cdot \nabla P(B-tA)| \diff t \diff B \\
&= \int_0^1 (\int |A\cdot \nabla P(B-tA)| \diff B) \diff t \\
&= \int_0^1 \diff t \int |A\cdot \nabla P(Y)| \diff Y.
\end{split}
\end{equation}
The integral over $t$ now evaluates to $1$ and we have proved the result.
The extension to all $L^1$ functions follows from the density of smooth
functions in $L^1$.

Next we prove~\eqref{eq:dilation-invariance}.  Fixing $B$, define
\[
f_B(t) = |1+t\delta|^{-2N} P((1+t\delta)^{-1}B).
\]
Then we compute
\begin{align*}
f_B'(t) = &-2N |1+t\delta|^{-2N-2}\Rept(\delta^* (1+\delta t)) P((1+t\delta)^{-1}B) \\
&\qquad + |1+t\delta|^{-2N} (-\delta)(1+t\delta)^{-2}B\cdot \nabla P((1+t\delta)^{-1}B).
\end{align*}
In particular, using the bound $|\delta|\leq 1/10$, we obtain
\[
|f_B'(t)| \leq 10|1+t\delta|^{-2N} \delta ( NP((1+t\delta)^{-1}B) + |B||\nabla P((1+t\delta)^{-1}B)|).
\]
Now we use the triangle inequality with the fundamental theorem of calculus as before to estimate
\begin{align*}
\int |P(B)-D_{1+\delta}P(B)|\diff B
&= \int |f_B(0) - f_B(1)|\diff B  \\
&\leq \int_0^1 \int |f_B'(t)|\diff B \diff t \\
&\leq 10\delta  \int_0^1 \int |1+t\delta|^{-2N} [N|P((1+t\delta)^{-1}B)| + |B||\nabla P((1+t\delta)^{-1}B)]\diff B \diff t \\
&\leq 10\delta  \int_0^1 \int N|P(B)| + |B||\nabla P(B)]\diff B \diff t.
\end{align*}
In the last line we rescaled the variable $B$, which eliminated the factor of $|1+t\delta|^{-2N}$.
This concludes the proof.
\end{proof}

Having expressed translation and dilation of probability density $P$ in terms of an integral of $|\nabla P|,$ we bound the integral with an explicit formula in Lemma \ref{lem:gradient-bounds}.

\begin{lem}
\label{lem:gradient-bounds}
Let $P \in L^1(\Complex^{n\times n})$ be the probability density associated to the measure $\mu$,
and suppose $m\geq (2+\eps)n$.  Then there exists a constant $C_\eps$ such that
\begin{equation}
\label{eq:BV-bd}
\int |\nabla P(B)|\diff B \leq C_\eps n
\end{equation}
and also
\[
\int |B| |\nabla P(B)| \diff B \leq C_\eps n^2.
\]
\end{lem}

Equation~\eqref{eq:BV-bd} shows that for
$\{0,1\}$-valued matrices $A$ (for which $\|A\|_{HS}\leq n$),
\[
\int |A\cdot \nabla P(B)|~dB\leq \int \|A\|_{HS} \|\nabla P(B)\|_{HS}~dB \leq C_\eps n^2.
\]
Thus the bounds from Lemma~\ref{lem:gradient-bounds} combined with Lemma~\ref{lem:approximate-invariance} (for $N=n^2$
complete the proof of Theorem~\ref{thm:appx-invariance}.

The proof of Lemma \ref{lem:gradient-bounds} makes use of the following result about the largest singular values of submatrices
of Haar-random unitaries, which we prove in the next subsection.
\begin{lem}
\label{lem:max-sing}
If $Z$ is the $n\times n$ submatrix of a  Haar-random $m\times m$ unitary with $m = \lceil(2+\eps)n\rceil$
and $n\geq 2\eps^{-1}$, then
\begin{equation} 
\Expec (1-\sigma_{max}(Z))^{-2} \leq C_\eps.
\end{equation}
In fact, one can take $C_\eps = 10^{3\eps^{-1}}$.
\end{lem}

\begin{proof}[Proof of Lemma~\ref{lem:gradient-bounds}] 
    We first write down an explicit formula for the density $P$.
As computed by B. Collins in~\cite{collins2003integrales}\footnote{See also the thesis~\cite[Eq.~61 on p.~61]{reffy} and~\cite{Aaronson2013}.}
for $m > 2n$, the density $P$ for the measure $\mu$ 
is given by
\[
P(B) = c_{m,n}\prod_{i\in [n]} \left(1-\frac{\sigma_i^2}{m}\right)^{m-2n} I_{\sigma_i^2 \leq m},
\]
where $\sigma_i$ are the singular values of the matrix $B$.
A first step to computing $\nabla P$ is to first take the derivatives with respect to the
singular values which, for $m > 2n$,
\begin{align}
\partial_{\sigma_j} P(B)
&=
-c_{m,n}
2\sigma_j
\frac{(m-2n)}{m}\left(1-\frac{\sigma_j^2}{m}\right)^{m-2n-1} I_{\sigma_j^2\leq m} \notag
\prod_{i\not=j} \left(1-\frac{\sigma_i^2}{m}\right)^{m-2n} I_{\sigma_i^2\leq m} \\
&= -2\sigma_j \frac{m-2n}{m} \left(1-\frac{\sigma_j^2}{m}\right)^{-1} P(B).
\end{align}
We can use this to bound $|\nabla P(B)|$.  To do this, we use a singular value decomposition
to reduce to the case that $B$ is the diagonal matrix $\Sigma$ with entries $\sigma_i$.  The derivative
of the singular values with respect to the off-diagonal entries is $0$ (it suffices to check this fact
for $2\times 2$ matrices).  The same is true for the imaginary parts of the diagonal.  Therefore
\begin{equation}
\label{eq:nablaP-est}
\|\nabla P(B)\|
= \left(\sum_{j\in[n]} |\partial_{\sigma_j} P(B)|^2\right)^{1/2}.
\end{equation}
To simplify this expression we use $(1-\sigma_i^2/m)^{-1} \leq (1-\sigma_{max}^2/m)^{-1}$
in the computation of $\partial_{\sigma_i} P(B)$,
\[
\partial_{\sigma_j} P(B) \leq 2 \sigma_j\left(1-\frac{\sigma_{max}^2}{m}\right)^{-1} P(B),
\]
and therefore the Euclidean norm of the gradient can be bounded as follows
\[
|\nabla P(B)| \leq 2 \left(\sum_j \sigma_j^2\right)^{1/2} \left(1-\frac{\sigma_{max}^2}{m}\right)^{-1} P(B)
\]

Integrating over $B$ and then using the Cauchy-Schwarz inequality we have 
\[
\int |\nabla P(B)|\diff B
= \int \frac{|\nabla P(B)|}{P(B)} P(B)\diff B
= \Expec\left(\frac{|\nabla P(B)|}{P(B)}\right) 
\leq 2 \left(\Expec \sum_j \sigma_j^2\right)^{1/2}\left(\Expec (1-\frac{\sigma_{max}^2}{m})^{-2}\right)^{1/2}
\]
where the expectation is over $B$ sampled from $\mu$.  Note that for the first expectation
$\sum_j \sigma_j^2 = |B|^2$ is the square of the Frobenius norm.  
Therefore
\[
\Expec \sum_j \sigma_j^2 = n^2,
\]
so
\[
\int |\nabla P| \diff B \leq 2 n \left(\Expec \left(1-\frac{\sigma_{max}^2}{m}\right)^{-2}\right)^{1/2}.
\]
Noting that $\sigma_{max}/\sqrt{m}$ has the law of the max singular value of $Z$ (the unscaled submatrix of a Haar-random unitary), 
Lemma~\ref{lem:max-sing} implies that the expectation above is bounded by some $C_\eps$.  Therefore for $C'_\eps = 2\sqrt{C_\eps}$,
\begin{equation}
\int |\nabla P|\diff B \leq C'_\eps n.
\end{equation}

To estimate $\int |B||\nabla P|$, we use that $|B|^2=\sum \sigma_j^2$ and therefore
\[
|B||\nabla P| \leq 2 |B|^2 \left(1-\frac{\sigma_{max}^2}{m}\right)^{-1} P(B).
\]
Integrating and using the Cauchy-Schwarz inequality it follows that
\[
\int |B||\nabla P| \diff B \leq 2 \left(\Expec |B|^4\right)^{1/2} \left(\Expec (1-\frac{\sigma_{max}^2}{m})^{-2}\right)^{1/2}
\leq C''_\eps n^2.
\]
This completes the proof of Lemma~\ref{lem:gradient-bounds}.
\end{proof}

\subsection{Bound on the maximum singular value of a submatrix}
In this subsection we prove Lemma~\ref{lem:max-sing}
 on the expectation $\Expec (1-\sigma(Z))^{-2}$ of a submatrix
 of a Haar-random unitary.  

In fact we will focus instead on the following  bound for the probability that the maximum singular value is close to $1$.

\begin{lem}
\label{lem:prob-max-sing}
Let $Z$ be the $n\times n$ upper left square submatrix of a
Haar-random $m\times m$ unitary matrix, and let $\delta < 1/2$.  Then
\[
\Prob(\sigma_{max}(Z)\geq 1-\delta) < 2^{5m+n+1} \delta^{2m-4n}.
\]
\end{lem}

Before proving Lemma~\ref{lem:prob-max-sing}, we show how it implies Lemma~\ref{lem:max-sing}.

\begin{proof}[Proof of Lemma~\ref{lem:max-sing} using Lemma~\ref{lem:prob-max-sing}]
Using the layer-cake formula,
\begin{equation}
\Expec (1-\sigma_{max}(Z))^{-2} =
\int_0^\infty \Prob( (1-\sigma_{max}(Z))^{-2} \geq t)\diff t.
\end{equation}
For $0\leq t\leq K$ we use the fact that the integrand is bounded by $1$.  Then, upon rearranging and
then applying Lemma~\ref{lem:prob-max-sing} we have
\begin{equation}
\begin{split}
\Expec (1-\sigma_{max}(Z))^{-2}
&\leq K + \int_K^\infty \Prob( (1-\sigma_{max}(Z))^{-2} \geq t)\diff t \\
&= K + \int_K^\infty \Prob( \sigma_{max}(Z) \geq  1- t^{-1/2})\diff t \\
&\leq K + 2^{5m+n+1} \int_K^\infty t^{2n-m} \diff t.
\end{split}
\end{equation}
We compute the latter integral using $m=\lceil(2+\eps)n\rceil$ and $n > 2\eps^{-1}$ to find
\begin{equation}
\begin{split}
\Expec (1-\sigma_{max}(Z))^{-2}
&\leq K + 2^{6m} \int_K^\infty t^{-\eps n} \diff t\\
&= K + 2^{6m} (\eps n-1)^{-1} K^{1-\eps n} \\
&\leq  K + 2^{6m} (\eps n/2)^{-1} K^{-\eps n/2}.
\end{split}
\end{equation}
We apply this bound with the choice $K = 500^{1/\eps}$ to conclude.
\end{proof}

\subsection{Interlude: facts about spheres}
In this section we use $\omega_n$ to denote the volume of a ball of radius $r$ in $n$ dimensions.  There is an
explicit formula for $\omega_n$ in terms of the Gamma function,
\begin{equation}
\label{eq:sphere-vol}
\omega_n = \frac{\pi^{n/2}}{\Gamma(n/2+1)}r^n.
\end{equation}
The surface area of an $(n-1)$-dimensional sphere in $n$ dimensions is given by $n\omega_nr^{n-1}$.

We give an elementary calculation regarding the surface area of the intersection of a sphere
and a ball.
\begin{lem}
\label{lem:cap-volume}
Let $w_0\in\Sphere^{2n-1}$ be any unit vector and let $\delta<1$.  Then
\[
\vol_{2n-1} \{w\in\Sphere^{2n-1} \mid \|w-w_0\|\leq \delta\}
\geq \omega_{2n-1} 2^{-n-1/2} \delta^{2n-1}.
\]
\end{lem}

\begin{figure}
\centering
\includegraphics{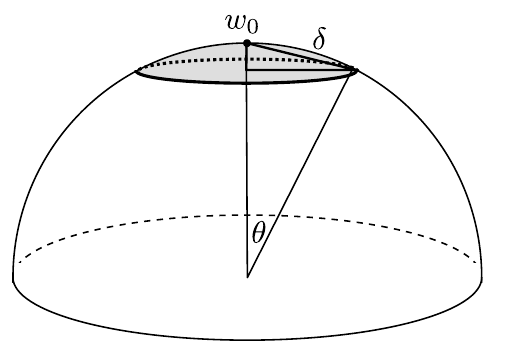} 
\caption{The spherical cap in a $\delta$ radius around $w_0$.  The shaded region is a spherical cap, whose boundary is a circle (in general, $d-2$-dimensional sphere).  The area of the cap is strictly larger than the area of the flat disk of radius $\sin\theta$ with the same boundary.  The value of $\theta$ is determined by the equation
$(1-\cos(\theta))^2 + (\sin\theta)^2 = \delta$, which can be seen by inspecting the right triangle drawn with vertex at $w_0$ and hypotenuse of length $\delta$.}
\label{fig:sphere-fig}
\end{figure}

\begin{proof}
The set $\{w\in\Sphere^{2n-1} \mid \|w-w_0\|\leq\delta\}$ forms a spherical cap whose boundary is a
$2n-2$-dimensional sphere of radius $\sin(\theta)$, where $\theta$ is the angle satisfying
\[
\delta^2 = (1-\cos\theta)^2 + \sin^2\theta.
\]
See Figure~\ref{fig:sphere-fig} for a diagram of the spherical cap.
\
In terms of $r=\sin\theta$, which is the radius of the boundary of the spherical cap, we can write
\[
\delta^2 = (1-\sqrt{1-r^2})^2 + r^2.
\]
When $r<1$ (which is guaranteed by $\delta < 1$),
$(1-\sqrt{1-r^2}) < r$, so $\delta^2 \leq 2r^2$.  Therefore $r \geq \delta / \sqrt{2}$.

The volume of the spherical cap is larger than the volume of the flat disk of radius $r$ with the same boundary,
which concludes the proof.
\end{proof}

\begin{lem}
\label{lem:small-coords}
If $X=(x_1,x_2,\cdots,x_n)$ is sampled uniformly from the unit sphere, then
\[
\Prob(\sum_{j=1}^k x_j^2 \leq \delta) \leq 2^{n/2} \delta^k.
\]
\end{lem}
\begin{proof}
The first $k$ coordinates of $X$ are sampled from a probability distribution with density
\[
p_{k,n}(x_1,\dots,x_k) = \frac{\Gamma(n/2)}{\Gamma((n-k)/2)\pi^{k/2}} \big(1 - \sum_{j=1}^k x_j^2\big)^{(n-k)/2-1}.
\]
In particular the density is bounded by $\frac{\Gamma(n/2)}{\Gamma((n-k)/2) \pi^{k/2}}$, and the event that
$\sum_{j=1}^k x_j^2 \leq \delta$ is a ball of volume $\omega_k \delta^k$.  Using~\eqref{eq:sphere-vol}
to express $\omega_k$, we arrive at the bound
\[
\Prob(\sum_{j=1}^k x_j^2 \leq \delta) \leq \frac{\Gamma(n/2)}{\Gamma((n-k)/2)\Gamma(k/2+1)} \delta^k.
\]
Note that the prefactor involving the ratio of Gamma functions is the binomial coefficient $\binom{n/2-1}{k/2}$,
and is therefore bounded by $2^{n/2}$, which concludes the proof.
\end{proof}

\subsection{Proof of Lemma~\ref{lem:prob-max-sing}}

We think of the probability $\Prob(\sigma_{max}(Z) \geq 1-\delta)$ as the expectation of the indicator function
$\Expec \One_{\sigma_{max}(Z)\geq 1-\delta}$.  Though the indicator function is difficult to work with, fortunately
we can bound it from above by a simpler function $f(Z)$ given by
\[
f_\delta(Z) := \vol_{2n-1}\{w\in \Sphere^{2n-1} \mid \|Zw\| \geq 1-2\delta\}
= \int_{\Sphere^{2n-1}} \One_{\|Zw\| \geq 1-2\delta} \diff w.
\]
Here we think of an $n$-dimensional complex vector $w$ as a $2n$-dimensional real-valued vector.

We break the proof of Lemma~\ref{lem:prob-max-sing} into two parts.  In the first part, we show that
$f_\delta(Z)$ controls the indicator function $\One_{\sigma_{max}(Z)\geq 1-\delta}$.  In the second part,
we compute the expectation of $f_\delta(Z)$.

For the first part, our aim is to prove that, when $\delta < 1$,
\begin{equation}
\label{eq:indicator-vs-f}
f_\delta(Z) \geq \omega_{2n-1}2^{-n-1/2} \delta^{2n-1}
\One_{\sigma_{max}(Z) \geq 1-\delta}.
\end{equation}
This is equivalent to showing that, when $\sigma_{max}(Z)\geq 1-\delta$,
\[
\vol_{\Sphere^{2n-1}}\{w\in \Sphere^{2n-1} \mid \|Zw\| \geq 1-2\delta\}
\geq \omega_{2n-1}2^{-n-1/2} \delta^{2n-1}.
\]
Now suppose that $\sigma_{max}(Z)\geq 1-\delta$, so that in particular there exists a unit
vector $w_0$ such that $\|Zw_0\| \geq 1-\delta$.  Then, for unit vectors $w$ such that $\|w-w_0\|\leq \delta$,
\[
\|Zw\| \geq \|Zw_0\| - \|Z(w-w_0)\| \geq 1-2\delta.
\]
Therefore
\[
\vol_{\Sphere^{2n-1}}\{w\in \Sphere^{2n-1} \mid \|Zw\| \geq 1-2\delta\}
\geq
\vol_{\Sphere^{2n-1}}\{w\in \Sphere^{2n-1} \mid \|w-w_0\|\leq \delta\}.
\]
Now~\eqref{eq:indicator-vs-f} follows from Lemma~\ref{lem:cap-volume}.

For the second part of the proof we bound $\Expec f_\delta(Z)$.  Using rotational symmetry
we have
\begin{equation}
\begin{split}
\Expec f_\delta(Z) &= \int_{\Sphere^{2n-1}} \Prob(\|Zw\|\geq 1-2\delta) \diff w \\
&= \omega_{2n-1} \Prob(\|Ze\| \geq 1-2\delta),
\end{split}
\end{equation}
where $e=(1,0,\cdots,0)$.  The random vector $Ze$ consists of the first $n$ coordinates of the first column of a
Haar-random unitary, which is simply a uniformly random vector in the unit sphere.  Therefore
\[
\Prob(\|Ze\|\geq 1-2\delta) = \Prob(\sum_{j=1}^{2n} x_j^2 \geq (1-2\delta)^2).
\]
when $(x_1,\cdots,x_{2m})$ is sampled from the unit sphere.  Rewriting the latter probability as the
event that the final $2m-2n$ coordinates are small and then applying Lemma~\ref{lem:small-coords} we find
\[
\Prob(\sum_{j=1}^{2n} x_j^2 \geq (1-2\delta)^2)
= \Prob(\sum_{j=2n+1}^{2m} x_j^2 \leq 4(\delta - \delta^2))
\leq 32^{m-n} \delta^{2m-2n}.
\]

Combining our two bounds on $f_\delta(Z)$ we have
\begin{equation}
\omega_{2n-1} 2^{-n-1/2} \delta^{2n-1} \Prob(\sigma_{max}(Z) \geq 1-\delta)
\leq
\Expec f_\delta(Z)
\leq \omega_{2n-1} 32^{m-n} \delta^{2m-2n}.
\end{equation}
Lemma~\ref{lem:prob-max-sing} now follows from rearranging.

\end{document}